\documentclass[12pt,a4paper]{amsart}

\usepackage[latin1]{inputenc}
\usepackage{amsthm}
\usepackage{latexsym}
\usepackage{amsfonts}
\usepackage{bbm,bm,dsfont}
\usepackage{color} 
\usepackage{graphicx}

\newtheorem{theorem}{Theorem}
\newtheorem{lemma}{Lemma}

\newtheorem{proposition}{Proposition}

\theoremstyle{definition}
\newtheorem{remark}{Remark}

\newtheorem{definition}{Definition}

\newcommand{\mc}[1]{\mathcal{#1}}
\newcommand{\ms}[1]{\mathsf{#1}}
\newcommand{\mf}[1]{\mathfrak{#1}}
\newcommand{\mb}[1]{\mathbb{#1}}

\newcommand{\N}{\mathbb N}
\newcommand{\R}{\mathbb R}

\newcommand{\C}{\mathbb C}

\newcommand{\T}{\mathbb T}

\newcommand{\hil}{\mathcal{H}} %Hilbert space

\newcommand{\tr}[1]{\mathrm{tr}\left[#1\right]} %trace
\def\<{\langle}
\def\>{\rangle}

\newcommand{\id}{\mathbbm{1}} %identity operator
\newcommand{\fii}{\varphi}

\newcommand{\eps}{\varepsilon}
\newcommand{\tj}{\vartheta}
\newcommand{\ovl}{\overline}
\newcommand{\vN}{\otimes_{\rm vN}}
\newcommand{\alg}{\otimes_{\rm alg}}
\newcommand{\tmin}{\otimes_{\rm min}}
\newcommand{\tmax}{\otimes_{\rm max}}

\begin{document}

\title{The Choi-Jamio\l kowski isomorphism\\
and covariant quantum channels}

\author{Erkka Haapasalo}
\address{Department of Physics and Center for Field Theory and Particle Physics, Fudan University, Shanghai 200433, China}
\email{erkka\_haapasalo@fudan.edu.cn}

\begin{abstract}
A generalization of the Choi-Jamio\l kowski isomorphism for completely positive maps between operator algebras is introduced. Particular emphasis is placed on the case of normal unital completely positive maps defined between von Neumann algebras. This generalization is applied especially to the study of maps which are covariant under actions of a symmetry group. We highlight with the example of, e.g.,\ phase-shift-covariant quantum channels the ease of this method in particular in the case of a compact symmetry group. We also discuss the case of channels which are covariant under actions of the Euclidean group of rigid motions in 3 dimensions.
\end{abstract}

\maketitle

\section{Introduction}

The (Choi-)Jamio\l kowski isomorphism \cite{Jamiolkowski} is an established and simple method in studying completely positive trace-preserving maps between finite quantum systems and it is a standard part of the quantum information researchers' tool kit. It simply identifies rank-1 operators $|n\>\<m|$ in a given orthonormal basis of a Hilbert space $\hil$ with vectors $|m,n\>:=|m\>\otimes|n\>$ of $\hil\otimes\hil$ and, thus, identifies state transformations with operators of the tensor product Hilbert space. This method has been earlier adapted for quantum channels between possibly infinite dimensional quantum systems \cite{Holevo2011,KiBuUoPe2017}, but we suggest a generalization of this which is applicable to a wider range of input and output operator algebras. Our main results deal with normal unital completely positive maps between injective von Neumann algebras where the Heisenberg-output algebra is, additionally, $\sigma$-finite.

In this treatise, our motivation for this generalization is the new methodology it provides for the study of covariant channels. The use of dilation techniques has earlier been used successfully to study covariant maps in infinite-dimensional cases \cite{HaPe2017}, but characterization of covariant quantum channels remains, in practice, rather difficult. The methods introduced here are a suggestion how to overcome these problems in many infinite-dimensional situations. We will see that the choice of the faithful state generating the generalized Choi-Jamio\l kowski isomorphism is crucial in simplifying the determination of covariant channels. Situation is particularly simple for compact symmetry groups.

This paper is arranged as follows: In Section \ref{sec:prel}, we present the basic definitions regarding operator algebras and channels between them and, in Section \ref{sec:CJ}, we put them to work and define the generalized Choi-Jamio\l kowski isomorphism (Definition \ref{def:ChoiJamiolkowski}). The prototypical case of channels between type-I factors (quantum-to-quantum channels) is studied in Subsection \ref{subsec:typeI} where a connection between particular Kraus decompositions of channels and spectral decompositions of their Choi states is established. Proposition \ref{prop:transposed} in Subsection \ref{subsec:transpose} shows that the Choi-Jamio\l kowski state of a channel is also the Choi-Jamio\l kowski state of a particular transposed channel up to swapping the places of the input and output algebras. We go on to covariance questions in Section \ref{sec:covCJ} and characterize the Choi-Jamio\l kowski states of covariant channels in Theorem \ref{theor:covCJ}. In Section \ref{sec:invariant} we discuss some examples where a faithful state $\rho_0$ on the Heisenberg-output algebra can be chosen so that $\rho_0$ is invariant under the symmetry action. This makes the characterization of covariant channels particularly simple as is exemplified by Subsection \ref{subsec:modulargroup} dealing with channels covariant with respect to modular automorphism groups and Subsection \ref{subsec:phaseshift} discussing quantum channels which are covariant under phase shifts. As a further example, we investigate channels which are covariant with respect to rigid motions in $\R^3$.

\section{Basic definitions and mathematical preliminaries}\label{sec:prel}

This section fixes some basic notations used throughout this work and also introduces mathematical concepts and results paving the way for the establishment of the Choi-Jamio\l kowski isomorphism in the following section. Recall that a $C^*$-algebra $\mf M$ is a von Neumann algebra if it is the topological dual of a Banach space known as the {\it pre-dual}. The pre-dual is unique up to homeomorphism and is denoted $\mf M_*$. We may, equivalently, define a von Neumann algebra as a sub-$C^*$-algebra of $\mc L(\mf H)$ coinciding with its double commutant, where $\mc L(\mf H)$ stands for the algebra of bounded operators on a Hilbert space $\mf H$. We denote the unit of $\mf M$ by $1_{\mf M}$ and the unit of $\mc L(\mf H)$ by $\id_{\mf H}$. The elements of the pre-dual of a von Neumann algebra $\mf M$ are viewed as functionals on $\mf M$ through $\rho(a)=\<\rho,a\>$, $\rho\in\mf M_*$, $a\in\mf M$.

Let $\mc A$ and $\mc B$ be $C^*$-algebras. Pick $n\in\N$. We say that a linear map $\Phi:\mc B\to\mc A$ is {\it $n$-positive} if, for every $a_1,\ldots,\,a_n\in\mc A$ and all $b_1,\ldots,\,b_n\in\mc B$,
$$
\sum_{i,j=1}^n a_i^*\Phi(b_i^*b_j)a_j\geq0.
$$
This $\Phi$ is {\it completely positive} if it is $n$-positive for all $n\in\N$. If $\mc A$ is Abelian, complete positivity reduces to positivity (i.e.,\ 1-positivity). We define complete positivity similarly for general *-algebras $\mc A$ and $\mc B$. If $\mc A$ and $\mc B$ possess units $1_{\mc A}$ and, respectively $1_{\mc B}$, a map $\Phi:\mc B\to\mc A$ is {\it unital} if $\Phi(1_{\mc B})=1_{\mc A}$.
\begin{definition}
Whenever $\mc A$ and $\mc B$ are unital $C^*$-algebras, we call unital completely positive linear maps $\Phi:\mc B\to\mc A$ as {\it channels} and denote the set of these channels by ${\bf CH}(\mc A,\mc B)$. The set of positive linear functionals on $\mc B$ is denoted by ${\bf S}(\mc B)$. If $\mc B$ is unital, we denote ${\bf S}_1(\mc B):={\bf CH}(\C,\mc B)$ and call these functionals as {\it states}.
\end{definition}

Suppose that $\mf M$ and $\mf N$ are von Neumann algebras. A positive map $\Phi:\mf N\to\mf M$ is continuous with respect to the ultraweak topologies of $\mf N$ and $\mf M$ if and only if it is {\it normal}: $\Phi\big(\sup_\lambda b_\lambda\big)=\sup_\lambda\Phi(b_\lambda)$ for any increasing net $(b_\lambda)_\lambda\subset\mf N$ of self-adjoint elements.
\begin{definition}
Whenever $\mf M$ and $\mf N$ are von Neumann algebras, we denote by ${\bf NCH}(\mf M,\mf N)$ the set of normal unital completely positive linear maps $\Phi:\mf N\to\mf M$. We call the maps $\Phi\in{\bf NCH}(\mf M,\mf N)$ as {\it normal channels}. We denote the set of normal positive unital functionals $\rho:\mf N\to\C$ by ${\bf NS}_1(\mf N)$. These states are called as {\it normal states}.
\end{definition}

Naturally, $\mf N_*$ coincides with the vector space spanned by normal states. The continuity properties of a channel $\Phi\in{\bf NCH}(\mf M,\mf N)$ guarantee the existence of the {\it pre-dual map} $\Phi_*:\mf M_*\to\mf N_*$, $\Phi_*(\rho)=\rho\circ\Phi$, $\rho\in\mf M_*$. This allows us to define normal channels equivalently as bounded linear maps $\Lambda:\mf M_*\to\mf N_*$ such that $\Lambda\big({\bf NS}_1(\mf M)\big)\subseteq{\bf NS}_1(\mf N)$ and the dual $\Lambda^*:\mf N\to\mf M$ defined through $\rho\circ\Lambda^*=\Lambda(\rho)$, $\rho\in\mf M_*$, is completely positive.

Recall that, for any completely positive linear map $\Phi:\mc A\to\mc L(\mf H)$ where $\mc A$ is a $C^*$-algebra and $\mf H$ is a Hilbert space, there exists a {\it Stinespring dilation}, i.e.,\ a triple $(\mf K,\pi,J)$ consisting of a Hilbert space $\mf K$, a *-representation $\pi:\mc A\to\mc L(\mf K)$, and a linear map $J:\mf H\to\mf K$ such that $\Phi(a)=J^*\pi(a)J$, $a\in\mc A$. Amongst these dilations there exists a minimal $(\mf K_0,\pi_0,J_0)$ where the vectors $\pi(a)J\fii$, $a\in\mc A$, $\fii\in\mf H$, span a dense subspace of $\mf K$. The minimal dilation is unique up to unitary equivalence. If $\mc A=\mf N$ is a von Neumann algebra, the above $\Phi$ is normal if and only if $\pi$ is a normal *-representation. The notion of the Stinespring dilation can be extended to the case where we replace the (Heisenberg) output algebra $\mc L(\mf H)$ with a more general $C^*$-algebra, but that involves the use of $C^*$-modules which is irrelevant for our scope.

From the above we obtain the GNS-constructions of states: A pair $(\mf H,\Omega)$ consisting of a Hilbert space $\mf H$ and a unit vector $\Omega\in\mf H$ is a {\it GNS-construction} for a state $\rho\in\mc S(\mf M)$ if $\mf M\subseteq\mc L(\mf H)$ and $\rho(a)=\<\Omega|a\Omega\>$ for all $a\in\mf M$. Such a GNS construction exists for any state and among them there is a minimal one for which $\Omega$ is cyclic for $(\mf M,\mf H)$, i.e.,\ $a\Omega$, $a\in\mf M$, span a dense subspace of $\mf H$.

For two von Neumann algebras $\mf M$ and $\mf N$ operating on the respective Hilbert spaces $\mf H$ and $\mf K$ we can define the {\it von Neumann tensor product} $\mf M\vN\mf N$ as the double commutant of the algebraic tensor product $\mf M\alg\mf N$ in $\mc L(\mf H\otimes\mf K)$. However, we need some more generalized tensor products in this treatise: Suppose that $\mc A$ and $\mc B$ are $C^*$-algebras. A {\it cross norm for $\mc A$ and $\mc B$} is a norm $\|\cdot\|:\mc A\alg\mc B\to[0,\infty)$ such that $\|cd\|\leq\|c\|\|d\|$ and $\|c^*c\|=\|c\|^2$ for all $c,\,d\in\mc A\alg\mc B$. The completion of $\mc A\alg\mc B$ with respect to any cross norm is a $C^*$-algebra. There are the {\it minimal} and, respectively, {\it maximal cross norms} $\|\cdot\|_{\rm min}$ and, respectively, $\|\cdot\|_{\rm max}$ defined by
\begin{eqnarray*}
\bigg\|\sum_{j=1}^na_j\otimes b_j\bigg\|_{\rm min}&=&\sup_{\pi_1\in{\rm Repr}\,\mc A,\,\pi_2\in{\rm Repr}\,\mc B}\bigg\|\sum_{j=1}^n\pi_1(a_j)\otimes\pi_2(b_j)\bigg\|,\\
\bigg\|\sum_{j=1}^na_j\otimes b_j\bigg\|_{\rm max}&=&\sup_{\pi\in{\rm Repr}\,\mc A\alg\mc B}\bigg\|\pi\bigg(\sum_{j=1}^na_j\otimes b_j\bigg)\bigg\|
\end{eqnarray*}
for all $n\in\N$, $a_1,\ldots,\,a_n\in\mc A$, and $b_1,\ldots,\,b_n\in\mc B$, where ${\rm Repr}\,\mc C$ for a *-algebra $\mc C$ stands for the class of *-representations of $\mc C$. For any cross norm $\|\cdot\|$ and any $c\in\mc A\alg\mc B$, $\|c\|_{\rm min}\leq\|c\|\leq\|c\|_{\rm max}$. The $\|\cdot\|_{\rm min}$-completion of $\mc A\alg\mc B$ is denoted $\mc A\tmin\mc B$ and the $\|\cdot\|_{\rm max}$-completion is denoted $\mc A\tmax\mc B$. For von Neumann algebras $\mf M\subseteq\mc L(\mf H)$ and $\mf N\subseteq\mc L(\mf K)$ the minimal cross norm $\|\cdot\|_{\rm min}$ coincides with the operator norm of $\mc L(\mf H\otimes\mf K)$ restricted on $\mf M\alg\mf N$.

A von Neumann algebra $\mf M$ is {\it injective} if, for every unital $C^*$-algebra $\mc A$ and any selfadjoined closed subspace $\mc V$ of $\mc A$ containing the unit of $\mc A$, and for any completely positive linear map $\Psi:\mc V\to\mf M$, there is a completely positive linear map $\ovl{\Psi}:\mc A\to\mf M$ such that $\ovl{\Psi}|_{\mc V}=\Psi$. According to \cite[Chapter XV, Theorem 3.1]{TakesakiIII} $\mf M\subseteq\mc L(\mf H)$ is injective if and only if either of the following equivalent properties holds:
\begin{itemize}
\item[(I1)] Denote the commutant of $\mf M$ within $\mc L(\mf H)$ by $\mf M'$, i.e.,\ $\mf M'=\{a'\in\mc L(\mf H)\,|\,a'a=aa'\ {\rm for\ all}\ a\in\mf M\}$. For any $n\in\N$, $a_1,\ldots,\,a_n\in\mf M$, and $a'_1,\ldots,\,a'_n\in\mf M'$,
$$
\bigg\|\sum_{i=1}^na_ia'_i\bigg\|\leq\bigg\|\sum_{i=1}^na_i\otimes a'_i\bigg\|_{\rm min}.
$$
\item[(I2)] There are nets $(S_\lambda)_{\lambda\in L}$ and $(T_\lambda)_{\lambda\in L}$ of completely positive linear contractions $S_\lambda:\mf M\to\mc M_{n_\lambda}(\C)$, $T_\lambda:\mc M_{n_\lambda}(\C)\to\mf M$, where $(n_\lambda)_{\lambda\in L}$ is a net of natural numbers, $\mc M_n(\C)$ stands for the algebra of $n\times n$-matrices with complex entries for any $n\in\N$, and $S_\lambda$ is normal for any $\lambda\in L$ such that $T_\lambda\circ S_\lambda\overset{\lambda\in L}{\rightarrow}{\rm id}_{\mf M}$ in the point-ultrastrong topology.
\end{itemize}

\section{A generalization of the Choi-Jamio\l kowski isomorphism}\label{sec:CJ}

We now embark on defining a generalization for the traditional Choi-Jamio\l kowski isomorphism. We first make a necessary definition.

\begin{definition}\label{def:binormal}
\begin{itemize}
\item[(i)] Assume that $\mc A$ and $\mc B$ are unital $C^*$-algebras. Let $\mc A\otimes_x\mc B$ be any $C^*$-tensor product of $\mc A$ and $\mc B$, i.e., $x$ is anything between ${\rm min}$ and ${\rm max}$, or the von Neumann tensor product if $\mc A$ and $\mc B$ are von Neumann algebras. For a (normal) state $S\in{\bf S}(\mc A\otimes_x\mc B)$ (or $S\in{\bf NS}(\mc A\vN\mc B)$ in the von Neumann algebra case), we denote by $S_{(1)}\in{\bf S}(\mc A)$ (or $S\in{\bf NS}(\mc A)$) the state $\mc A\ni a\mapsto S(a\otimes 1_{\mc B})\in\C$ and by $S_{(2)}\in{\bf S}(\mc B)$ (or $S\in{\bf NS}(\mc B)$) the state $\mc B\ni b\mapsto S(1_{\mc A}\otimes b)\in\C$. We call $S_{(1)}$ and $S_{(2)}$ as the {\it margins of $S$}.
\item[(ii)] Suppose that $\mf M$ and $\mf N$ are von Neumann algebras and $\mf M\otimes_x\mf N$ is any of the tensor products of item (i) above. We denote by ${\bf S}^{\rm bin}(\mf M\otimes_x\mf N)$ the subset of those $S\in{\bf S}(\mf M\otimes_x\mf N)$ such that $S_{(1)}$ and $S_{(2)}$ are normal. We call these as the {\it binormal states}.
\end{itemize}
\end{definition}

Let us discuss item (ii) above in more detail: Let $S\in{\bf S}^{\rm bin}(\mf M\otimes_x\mf N)$, $(a_\lambda)_{\lambda\in L}$ an increasing net bounded from above and with the supremum $a$, and $b\in\mf N$ a fixed positive element. We have
$$
S\big((a-a_\lambda)\otimes b\big)\leq\|b\|S\big((a-a_\lambda)\otimes 1_{\mf N}\big)=\|b\|S_{(1)}(a-a_\lambda)\overset{\lambda\in L}{\searrow}0,
$$
implying that $S(a\otimes b)=\sup_{\lambda\in L}S(a_\lambda\otimes b)$; we have used the well known fact that $b\leq\|b\|1_{\mf N}$. It follows that $S(\cdot\otimes b)$ is normal for all positive $b\in\mf N$. Similarly, $S(a\otimes\cdot)$ is normal for all positive $a\in\mf M$. By expressing elements of von Neumann algebras as linear combinations of four positive elements, we have that $S\in{\bf S}(\mf M\otimes_x\mf N)$ is binormal if and only if $S(\cdot\otimes b)$ and $S(a\otimes\cdot)$ are ultraweakly continuous for all $a\in\mf M$ and $b\in\mf N$.

Before we can formulate the Choi-Jamio\l kowski isomorphism, we need to discuss some basics of modular theory. To this end, in turn, we have to introduce some further concepts dealing with faithful states. A von Neumann algebra $\mf M$ is {\it $\sigma$-finite} if any set of mutually orthogonal projections of $\mf M$ is countable. This $\sigma$-finiteness is equivalent with the existence of a faithful state $\rho_0\in\mc S(\mf M)$, i.e.,\ $\rho_0(a)=0$ for $a\in\mf M$, $a\geq0$, implies $a=0$. From now on, we assume that $\mf M$ is a $\sigma$-finite von Neumann algebra and $\mf N$ is another von Neumann algebra (not necessarily $\sigma$-finite) and we fix a faithful state $\rho_0\in\mc S(\mf M)$. We fix a GNS-construction $(\mf H,\Omega)$ for $\rho_0$ where $\Omega\in\mf H$ is cyclic and separating vector for $(\mf H,\mf M)$; the latter condition means that the map $\mf M\ni a\mapsto a\Omega\in\mf H$ is injective and mirrors the fact that $\rho_0$ is faithful. Whenever $\mf M$ is $\sigma$-finite and $\rho_0\in{\bf NS}_1(\mf M)$ is faithful, a GNS construction $(\mf H,\Omega)$ for $\rho_0$ can always be found where $\Omega$ is cyclic and separating for $(\mf H,\mf M)$. We again denote the commutant of $\mf M$ in $\mc L(\mf H)$ by $\mf M'$. It can be shown that $\Omega$ is cyclic for $(\mf H,\mf M)$ if and only if it is separating for $(\mf H,\mf M')$ and $\Omega$ is separating for $(\mf H,\mf M)$ if and only if it is cyclic for $(\mf H,\mf M')$. Hence, the situation is completely symmetric for $\mf M$ and $\mf M'$, and we may define the faithful state $\rho'_0\in{\bf NS}_1(\mf M')$, $\rho'_0(a')=\<\Omega|a'\Omega\>$, $a'\in\mf M'$.

We may define the closable antilinear operator $S$ defined densely by $Sa\Omega=a^*\Omega$, $a\in\mf M$. We may give $S$ the polar decomposition $S=J\Delta^{1/2}$, where $\Delta:=S^*S$ is a strictly positive operator called as the {\it modular operator} and $J=J^*=J^{-1}$ is an antilinear isometry called as the {\it modular conjugation}. The Tomita-Takesaki modular theorem \cite[Chapter VI, Theorem 1.19]{TakesakiII} states that
$$
J\mf MJ=\mf M',\quad \Delta^{it}\mf M\Delta^{-it}=\mf M,\quad\Delta^{it}\mf M'\Delta^{-it}=\mf M',\qquad t\in\R.
$$
We denote $j(C):=JCJ$ for all $C\in\mc L(\mf H)$. We have that $\rho'_0(a')=\ovl{(\rho_0\circ j)(a')}$ for all $a'\in\mf M'$.

The following theorem serves to establish the theoretical basis of the Choi-Jamio\l kowski isomorphism.

\begin{theorem}\label{theor:CJisomorphism}
Suppose that $\mf M$ is an injective $\sigma$-finite von Neumann algebra and retain the notations above. Also assume that $\mc B$ is a unital $C^*$-algebra. Define ${\bf S}_{\rho_0}(\mf M'\tmin\mc B)$ as the set of those $S\in{\bf S}_1(\mf M'\tmin\mc B)$ such that $S_{(1)}=\rho'_0$. For each $\Phi\in{\bf CH}(\mf M,\mc B)$, define $S_\Omega^\Phi:\mf M'\alg\mc B\to\C$, $S_\Omega^\Phi(a'\otimes b)=\<\Omega|a'\Phi(b)\Omega\>$, $a'\in\mf M'$, $b\in\mc B$. The map $S_\Omega^\Phi$ extends into a state $S_\Omega^\Phi\in{\bf S}_{\rho_0}(\mf M'\tmin\mc B)$ and the map
$$
{\bf CH}(\mf M,\mc B)\ni\Phi\mapsto S_\Omega^\Phi\in{\bf S}_{\rho_0}(\mf M'\tmin\mc B)
$$
is bijective. If, additionally, $\mc B=\mf N$ is a von Neumann algebra as well and we define ${\bf S}^{\rm bin}_{\rho_0}(\mf M'\tmin\mf N)$ as the set of those binormal states $S$ on $\mf M'\tmin\mf N$ such that $S_{(1)}=\rho'_0$, the above construction gives a bijection
$$
{\bf NCH}(\mf M,\mf N)\ni\Phi\mapsto S_\Omega^\Phi\in{\bf S}^{\rm bin}_{\rho_0}(\mf M'\tmin\mf N).
$$
\end{theorem}

\begin{proof}
Let $\mf M$ be an injective $\sigma$-finite von Neumann algebra and $\mc B$ a unital $C^*$-algebra as in the claim. According to item (I1) in the definition of injectivity, the map
$$
\mf M'\alg\mf M\ni\sum_{j=1}^na'_i\otimes a_i\mapsto\sum_{i=1}^na'_ia_i\in\mc L(\mf H)
$$
extends into a (unique) continuous homomorphism defined on $\mf M'\tmin\mf M$. We denote this homomorphism by $\tj$. Moreover, for any $\Phi\in{\bf CH}(\mf M,\mc B)$, the map ${\rm id}_{\mf M'}\alg\Phi:\mf M'\alg\mc B\to\mf M'\alg\mf M$, $a'\otimes b\mapsto a'\otimes\Phi(b)$, extends into a unique completely positive unital map ${\rm id}_{\mf M'}\tmin\Phi:\mf M'\tmin\mc B\to\mf M'\tmin\mf M$ \cite[Chapter IV, Proposition 4.23]{TakesakiI}. Thus we may define a completely positive unital map $\tj\circ({\rm id}_{\mf M'}\tmin\Phi):\mf M'\tmin\mc B\to\mc L(\mf H)$ such that $[\tj\circ({\rm id}_{\mf M'}\tmin\Phi)](a'\otimes b)=a'\Phi(b)$, $a'\in\mf M'$, $b\in\mc B$. Now, for any $\Phi\in{\bf CH}(\mf M',\mc B)$, $\mf M'\times\mc B\ni(a',b)\mapsto a'\Phi(b)\in\mc L(\mf H)$ extends into a unique channel within ${\bf CH}\big(\mc L(\mf H),\mf M'\tmin\mc B\big)$. Thus, $S_\Omega^\Phi\in{\bf S}(\mf M'\tmin\mc B)$ is well defined. It is immediately seen that, in fact, $S_\Omega^\Phi\in{\bf S}_{\rho_0}(\mf M'\tmin\mf N)$. If $\Phi,\,\Phi'\in{\bf CH}(\mf M,\mc B)$ and $\Phi\neq\Phi'$, the fact that $\Omega$ is cyclic and separating for $(\mf M,\mf H)$ easily implies that $S_\Omega^\Phi\neq S_\Omega^{\Phi'}$.

Suppose now that $S\in\mc S_{\rho_0}(\mf M'\tmin\mc B)$. Pick $b\in\mc B$, $b\geq0$, and define the positive map $S_b:\mf M'\to\C$ ($S_b\in\mf M'_*$), $S_b(a')=S(a'\otimes b)$, $a'\in\mf M'$. Let $(\mf K,\Psi)$ be a GNS-construction for $S_b$, where $\|\Psi\|^2=S(1_{\mf M}\otimes b)$. Pick $a'\in\mf M'$. We may evaluate
$$
\|a'\Psi\|^2=S_b((a')^*a')=S((a')^*a'\otimes b)\leq\|b\|S((a')^*a'\otimes1_{\mf N})=\|b\|\rho_0((a')^*a')=\|b\|\|a'\Omega\|^2
$$
which together with the fact that $\Omega$ is cyclic for $(\mf H,\mf M')$ implies that we may define a bounded operator $D_b:\mf H\to\mf K$ with $\|D_b\|\leq\sqrt{\|b\|}$ such that $D_b a'\Omega=a'\Psi$ for all $a'\in\mf M'$. Define $\Phi(b):=D_b^*D_b\in\mc L(\mf H)$. We find that, for any $a',\,\tilde{a}'\in\mf M'$, $\<\tilde{a}'\Omega|\Phi(b)a'\tilde{a}'\Omega\>=\<\tilde{a}'\Psi|a'\tilde{a}'\Psi\>=\<(a')^*\tilde{a}'\Psi|\tilde{a}'\Psi\>=\<D_b(a')^*\tilde{a}'\Omega|D_b\tilde{a}'\Omega\>=\<\tilde{a}'\Omega|a'\Phi(b)\tilde{a}'\Omega\>$. Thus, using the $(\mf H,\mf M')$-cyclicity of $\Omega$, we find for every positive $b\in\mf N$ a unique $\Phi(b)\in\mf M$ (which is positive) such that
\begin{equation}\label{eq:apu1}
S(a'\otimes b)=\<\Omega|a'\Phi(b)\Omega\>,\qquad a'\in\mf M'.
\end{equation}
Through linear extension, we may uniquely define a linear map $\Phi:\mc B\to\mf M$ such that \eqref{eq:apu1} is satisfied with all $b\in\mc B$. Pick $n\in\N$, $a'_1,\ldots,\,a'_n\in\mf M'$, and $b_1,\ldots,\,b_n\in\mc B$. We now have
$$
\sum_{i,j=1}^n\<a'_i\Omega|\Phi(b_i^*b_j)a'_j\Omega\>=\sum_{i,j=1}^n S\big((a'_i\otimes b_i)^*(a'_j\otimes b_j)\big)\geq0,
$$
implying that $\Phi$ is completely positive. The $(\mf H,\mf M')$-cyclicity of $\Omega$ can also be used to establish that $\Phi$ is unital, and we find that $S=S_\Omega^\Phi$.

Assume now that $\mc B=\mf N$ is a von Neumann algebra. As above, we may define $S_\Omega^\Phi\in{\bf S}_{\rho_0}(\mf M'\tmin\mf N)$ through $S_\Omega^\Phi(a'\otimes b)=\<\Omega|a'\Phi(b)\Omega\>$ for all $a'\in\mf M'$ and $b\in\mf N$. Let us establish the binormality of $S_\Omega^\Phi$: Naturally, $(S_\Omega^\Phi)_{(1)}=\rho'_0$ is normal. Let $(b_\lambda)_{\lambda\in L}$ be an increasing sequence of self-adjoint elements of $\mf N$ bounded from above and with the supremum $b$. Using the normality of $\Phi$, we have
$$
(S_\Omega^\Phi)_{(2)}(b_\lambda)=\<\Omega|\Phi(b_\lambda)\Omega\>\overset{\lambda\in L}{\nearrow}\<\Omega|\Phi(b)\Omega\>=(S_\Omega^\Phi)_{(2)}(b),
$$
implying the normality of $(S_\Omega^\Phi)_{(2)}$.

Suppose now that $S\in{\bf S}^{\rm bin}_{\rho_0}(\mf M'\tmin\mf N)$. As above, we find a unique $\Phi\in{\bf CH}(\mf M,\mf N)$ such that \eqref{eq:apu1} holds for every $b\in\mf N$. Assume that $(b_\lambda)_{\lambda\in L}\subset\mf N$ is an increasing net of self-adjoint elements bounded from above, $b_\lambda\nearrow b\in\mf N$. We have for all $a'\in\mf M'$
$$
\<a'\Omega|\Phi(b_\lambda)a'\Omega\>=S\big((a')^*a'\otimes b_\lambda\big)\nearrow S\big((a')^*a'\otimes b\big)=\<a'\Omega|\Phi(b)a'\Omega\>,
$$
implying that $\sup_\lambda\Phi(b_\lambda)=\Phi(b)$; we have used here the observations presented after Definition \ref{def:binormal}. Thus, $\Phi$ is normal and $\Phi\in{\bf NCH}(\mf M,\mf N)$ and we find that $S=S_\Omega^\Phi$.
\end{proof}

Following the beginning of the above proof and retaining the notations therein, we note that, whenever $\mf M$ is injective, we may define the state $\tilde{\rho_0}\in{\bf S}(\mf M'\tmin\mf M)$ as the extension of $\mf M'\times\mf M\ni(a',a)\mapsto\<\Omega|a'a\Omega\>\in\C$. Consider the special case where the state $\tilde{\rho_0}$ extends into a normal state $\hat{\rho_0}$ on $\mf M'\vN\mf M$. As we will see (and as is well-known) this happens especially when $\mf M$ is a type-I factor. It follows that in this case, for any von Neumann algebra $\mf N$ and any $\Phi\in{\bf NCH}(\mf M,\mf N)$, we may define the {\it normal} state $S_\Omega^\Phi\in{\bf NS}(\mf M'\vN\mf N)$ through $S_\Omega^\Phi(a'\otimes b)=\<\Omega|a'\Phi(b)\Omega\>$ for all $a'\in\mf M'$ and $b\in\mf N$ by combining the normal extension onto $\mf M'\vN\mf N$ of $\mf M'\times\mf N\ni(a',b)\mapsto a'\otimes\Phi(b)\in\mf M'\vN\mf M$ with $\hat{\rho_0}$. Exactly as in the above proof, the resulting map ${\bf NCH}(\mf M,\mf N)\ni\Phi\mapsto S_\Omega^\Phi\in{\bf NS}_{\rho_0}(\mf M'\vN\mf M)$, where ${\bf NS}_{\rho_0}(\mf M'\vN\mf N)$ is the set of those normal states $S$ on $\mf M'\vN\mf N$ such that $S_{(1)}=\rho'_0$, is bijective.

\begin{definition}\label{def:ChoiJamiolkowski}
Let $\mf M$ be an injective $\sigma$-finite von Neumann algebra and $\rho_0\in{\bf NS}_1(\mf M)$ be a faithful state with a GNS-construction $(\mf H,\Omega)$ where $\Omega$ is cyclic and separating for $(\mf H,\mf M)$ and retain the related notations fixed before Theorem \ref{theor:CJisomorphism}. Let $\mc B$ be a unital $C^*$-algebra. The bijection ${\bf CH}(\mf M,\mc B)\ni\Phi\mapsto S_\Omega^\Phi\in{\bf S}_{\rho_0}(\mf M'\tmin\mc B)$,
\begin{equation}\label{eq:ChoiJamiolkowski1}
S_\Omega^\Phi(a'\otimes b)=\<\Omega|a'\Phi(b)\Omega\>,\qquad\Phi\in{\bf CH}(\mf M,\mc B),\quad a'\in\mf M',\quad b\in\mc B,
\end{equation}
is called as the {\it Choi-Jamio\l kowski isomorphism} associated to the faithful state $\rho_0$ and its cyclic and separating GNS-construction $(\mf H,\Omega)$ and $S_\Omega^\Phi$ is called as the {\it Choi-Jamio\l kowski state of $\Phi$} for any $\Phi\in{\bf CH}(\mf M,\mc B)$.

If, in addition to the above assumptions, $\mc B=\mf N$ is a von Neumann algebra, the bijection ${\bf NCH}(\mf M,\mf N)\ni\Phi\mapsto S_\Omega^\Phi\in{\bf S}^{\rm bin}_{\rho_0}(\mf M'\tmin\mf N)$,
\begin{equation}\label{eq:ChoiJamiolkowski2}
S_\Omega^\Phi(a'\otimes b)=\<\Omega|a'\Phi(b)\Omega\>,\qquad\Phi\in{\bf NCH}(\mf M,\mf N),\quad a'\in\mf M',\quad b\in\mf N,
\end{equation}
is called as the {\it binormal Choi-Jamio\l kowski isomorphism} associated to the faithful state $\rho_0$ and its cyclic and separating GNS-construction $(\mf H,\Omega)$ and $S_\Omega^\Phi$ is called as the {\it binormal Choi-Jamio\l kowski state of $\Phi$} for any $\Phi\in{\bf NCH}(\mf M,\mf N)$.

Moreover, if the state $\tilde{\rho_0}$ presented after the proof of Theorem \ref{theor:CJisomorphism} extends into $\hat{\rho_0}\in{\bf NS}(\mf M'\vN\mf M)$, the bijection ${\bf NCH}(\mf M,\mf N)\ni\Phi\mapsto S_\Omega^\Phi\in{\bf NS}_{\rho_0}(\mf M'\vN\mf N)$ set up as in Equation \eqref{eq:ChoiJamiolkowski2} above is called as the {\it normal Choi-Jamio\l kowski isomorphism} associated to the faithful state $\rho_0$ and its cyclic and separating GNS-construction $(\mf H,\Omega)$ and $S_\Omega^\Phi$ is called as the {\it normal Choi-Jamio\l kowski state of $\Phi$} for any $\Phi\in{\bf NCH}(\mf M,\mf N)$.
\end{definition}

In addition to the above cases, one can also define in a sense weaker form of the Choi-Jamio\l kowski isomorphism where one does not have to assume $\mf M$ to be injective but the minimal tensor product has to be replaced with the maximal one. Recall that, for $C^*$-algebras $\mc A$, $\mc B$, and $\mc C$ and completely positive linear maps $\Phi_1:\mc A\to\mc C$ and $\Phi_2:\mc B\to\mc C$ with commuting ranges, one can uniquely define a completely positive map $\Phi:\mc A\tmax\mc B\to\mc C$ such that $\Phi(a\otimes b)=\Phi_1(a)\Phi_2(b)$ for all $a\in\mc A$ and $b\in\mc B$ \cite[Chapter IV, Proposition 4.23]{TakesakiI}. Suppose then that $\mf M$ is a $\sigma$-finite von Neumann algebras with a faithful normal state $\rho_0$ equipped with a GNS-representtion $(\mf H,\Omega)$ where $\Omega$ is cyclic and separating for $(\mf H,\mf M)$ and $\mc B$ is some $C^*$-algebra. For any $\Phi\in{\bf CH}(\mf M,\mc B)$ we may, according to the above define the completely positive extension of $\mf M'\times\mc B\ni(a',b)\mapsto a'\Phi(b)\in\mc L(\mf H)$ defined on $\mf M'\tmax\mc B$. The rest can be done in exactly the same way as in the proof of Theorem \ref{theor:CJisomorphism}, allowing the definition of a Choi-Jamio\l kowski isomorphism between ${\bf CH}(\mf M,\mc B)$ and ${\bf S}_{\rho_0}(\mf M'\tmax\mc B)$ and, similarly in the case where $\mc B=\mf N$ is a von Neumann algebra, the Choi-Jamio\l kowski isomorphism between ${\bf NCH}(\mf M,\mf N)$ and ${\bf S}^{\rm bin}_{\rho_0}(\mf M'\tmax\mf N)$. However, in what follows, we concentrate on the minimal tensor product and the fully normal cases although most of the theory of Section \ref{sec:covCJ} applies also to the maximal tensor product case.

Recall that, for real (or convex) vector spaces $V_i$ and convex sets $C_i\subseteq V_i$, $i=1,\,2$, a map $f:C_1\to C_2$ is {\it affine} if $f\big(tx+(1-t)y\big)=tf(x)+(1-t)f(y)$ for any $x,\,y\in C_1$ and $t\in[0,1]$. An affine bijection $f:C_1\to C_2$, in particular, translates the convex structure of $C_1$ into $C_2$ and vice-versa. Especially, for an affine bijection $f:C_1\to C_2$, $x\in C_1$ is an extreme point of $C_1$ if and only if $f(x)$ is an extreme point of $C_2$.

Let us assume that $\mf M$ and $\mf N$ are von Neumann algebras and $\mf M$ is, additionally, $\sigma$-finite and injective. Fix a faithful state $\rho_0\in{\bf NS}_1(\mf M)$ and let $(\mf H,\Omega)$ be a GNS-construction for $\rho_0$ where $\Omega$ is cyclic and separating for $(\mf H,\mf M)$. The set ${\bf NCH}(\mf M,\mf N)$ is convex; if $\Phi_1,\,\Phi_2\in{\bf NCH}(\mf M,\mf N)$ and $t\in[0,1]$, defining $t\Phi_1+(1-t)\Phi_2:\mf N\to\mf M$, $\big(t\Phi_1+(1-t)\Phi_2\big)(b)=t\Phi_1(b)+(1-t)\Phi_2(b)$, $b\in\mf N$, we have $t\Phi_1+(1-t)\Phi_2\in{\bf NCH}(\mf M,\mf N)$. Similarly, ${\bf S}^{\rm bin}_{\rho_0}(\mf M\tmin\mf N)$ is convex. It follows that the normal Choi-Jamio\l kowski isomorphism ${\bf NCH}(\mf M,\mf N)\ni\Phi\mapsto S_\Omega^\Phi\in{\bf S}^{\rm bin}_{\rho_0}(\mf M'\tmin\mf N)$ is an affine bijection implying that the isomorphism perfectly encodes the convex structures of ${\bf NCH}(\mf M,\mf N)$ onto ${\bf S}^{\rm bin}_{\rho_0}(\mf M'\tmin\mf N)$. Especially $\Phi\in{\bf NCH}(\mf M,\mf N)$ is an extreme point of ${\bf NCH}(\mf M,\mf N)$ if and only if $S_\Omega^\Phi$ is an extreme point of ${\bf S}^{\rm bin}_{\rho_0}(\mf M'\tmin\mf N)$. The same situation naturally holds in the case of the non-normal and fully normal Choi-Jamio\l kowski isomorphism.

\begin{remark}
One might wish to define the binormal Choi-Jamio\l kowski state of a $\Phi\in{\bf NCH}(\mf M,\mf N)$ on $\mf M\tmin\mf N$ without using the commutant $\mf M'$. One way of trying to do this is through the linear map $\mf M\ni a\mapsto a^T\in\mf M'$, $a^T=j(a)^*$, $a\in\mf M$. We may now define $T_\Omega^\Phi:\mf M\alg\mf N\to\C$, $T_\Omega^\Phi(a\otimes b)=S_\Omega^\Phi(a^T\otimes b)$ or, more explicitly,
\begin{eqnarray*}
T_\Omega^\Phi(a\otimes b)=\<\Omega|j(a)^*\Phi(b)\Omega\>=\<j(a)\Omega|\Phi(b)\Omega\>&=&\<Ja\Omega|\Phi(b)\Omega\>\\
&=&\<\Delta^{1/2}a^*\Omega|\Phi(b)\Omega\>=\<\Omega|a\Delta^{1/2}\Phi(b)\Omega\>
\end{eqnarray*}
for all $a\in\mf M$ and $b\in\mf N$. Note that the use of $j$ instead of $a\mapsto a^T$ would be problematic as $j$ is conjugate linear. However, the extension of $T_\Omega^\Phi$ into a state on $\mf M\tmin\mf N$ typically fails, as the map $a\mapsto a^T$ is not completely positive. This map is, indeed, related to the transpose with respect to a basis fixed by the choice of the GNS-construction for the faithful state in the case when $\mf M=\mc L(\hil)$ with some separable Hilbert space $\hil$ as will become clear in the sequel. Such a transpose is famously not completely positive.
\end{remark}

\subsection{Normal channels between type-I factors}\label{subsec:typeI}

In this subsection we concentrate on type-I factors $\mf M=\mc L(\hil)$ and $\mf N=\mc L(\mc K)$ where $\hil$ and $\mc K$ are Hilbert spaces and $\hil$ is, moreover, separable. We retain these assumptions throughout this subsection. We give an in-depth description of the Choi-states of the corresponding fully quantum channels working in the Heisenberg picture. In particular, we wil discuss Kraus decompositions of quantum channels and their connection to spectral decompositions of Choi-states. Much of these results are established with future applications in mind.

Recall that we may identify states of a type-I factor with the set of positive trace-class operators of trace 1. This set will be denoted by $\mc S(\hil)$. Faithful states $\rho_0\in\mc S(\hil)$ correspond to injective state operators. Let us fix such a state $\rho_0=\sum_{\xi\in K}t_\xi|\xi\>\<\xi|$ where $K\subset\hil$ is an orthonormal basis (which is countable), $t_\xi>0$ for all $\xi\in K$, and $\sum_{\xi\in K}t_\xi=1$. We may choose the minimal GNS-construction $(\hil\otimes\hil,\Omega)$ for $\rho_0$ where $\Omega=\sum_{\xi\in K}\sqrt{t_\xi}\xi\otimes\xi$. When defining the Choi-Jamio\l kowski isomorphism, we may choose more general minimal GNS-vectors $\Omega$ for $\rho_0$, but the above choice is the simplest and usually a different choice does not matter much. Note however that equations \eqref{eq:recovery1} and \eqref{eq:recovery2} below hold usually only if we choose $\Omega=\sum_{\xi\in K}\sqrt{t_\xi}\xi\otimes\xi$. Note that $\mc L(\hil)$ operates on $\hil\otimes\hil$ in this context in the form $\mc L(\hil)\times\hil\otimes\hil\ni(A,\eta)\mapsto(A\otimes\id_\hil)\eta\in\hil\otimes\hil$. Hence, $\mc L(\hil)'$ within $\mc L(\hil\otimes\hil)$ is $\{\id_\hil\otimes A'\,|\,A'\in\mc L(\hil)\}$, i.e.,\ $\mc L(\hil)'\simeq\mc L(\hil)$ and the von Neumann algebra generated in $\mc L(\hil\otimes\hil)$ can thus be identified with $\mc L(\hil)'\vN\mc L(\hil)$. The latter observation yields that the state $\tilde{\rho_0}$ extends into the normal state $\hat{\rho_0}$ with the density operator $|\Omega\>\<\Omega|$. Thus, we have access to the normal Choi-Jamio\l kowski isomorphism. Naturally, ${\bf NS}_{\rho_0}\big(\mc L(\hil)'\vN\mc L(\mc K)\big)$ is identified with the set $\mc S_{\rho_0}(\hil\otimes\mc K)$ of those $S\in\mc S(\hil\otimes\mc K)$ whose first partial trace ${\rm tr}_{\mc K}[S]=\rho_0$.

Let $\Phi\in{\bf NCH}\big(\mc L(\hil),\mc L(\mc K)\big)=:{\bf NCH}(\hil,\mc K)$. In this setting, we have, for all $A'\in\mc L(\hil)$ and $B\in\mc L(\mc K)$,
$$
\tr{S_\Omega^\Phi(A'\otimes B)}=\<\Omega|\big(\Phi(B)\otimes A'\big)\Omega\>=\tr{(\Phi_*\otimes{\rm id}_{\mc T(\hil)}(|\Omega\>\<\Omega|)(B\otimes A')},
$$
where $\Phi_*\otimes{\rm id}_{\mc T(\hil)}:\,\mc T(\hil\otimes\hil)\to\mc T(\mc K\otimes\hil)$ is the extended map defined by $(\Phi_*\otimes{\rm id}_{\mc T(\hil)})(T_1\otimes T_2)=\Phi_*(T_1)\otimes T_2$, $T_1,\,T_2\in\mc T(\hil)$. Thus,\ $S_\Omega^\Phi=U_{\rm SWAP}\big(\Phi_*\otimes{\rm id}_{\mc T(\hil)}(|\Omega\>\<\Omega|)\big)U_{\rm SWAP}^*$, where $U_{\rm SWAP}:\mc K\otimes\hil\to\hil\otimes\mc K$ is the unitary defined through $U_{\rm SWAP}\psi\otimes\fii=\fii\otimes\psi$ for all $\fii\in\hil$ and $\psi\in\mc K$. Hence, we recover the traditional Choi-Jamio\l kowski isomorphism. The swap unitary $U_{\rm SWAP}$ appears only because of aesthetic reasons; we want $\mf M'$ to appear before $\mf N$ in Definition \ref{def:ChoiJamiolkowski} in order to conserve alphabetical order. If the order was $\mf N\vN\mf M'$ in Definition \ref{def:ChoiJamiolkowski}, $U_{\rm SWAP}$ would vanish here. Hence, $U_{\rm SWAP}$ is inessential.

Let $\Phi\in{\bf NCH}(\hil,\mc K)$ and denote the transpose of $A\in\mc L(\hil)$ with respect to the basis $K$ by $A^T$, i.e.,\ $\<\xi|A^T\zeta\>=\<\zeta|A\xi\>$ for all $\zeta,\,\xi\in K$. We have, for all $A'\in\mc L(\hil)$ and $B\in\mc L(\mc K)$,
\begin{eqnarray}
\tr{S_\Omega^\Phi(A'\otimes B)}&=&\<\Omega|\big(\Phi(B)\otimes A'\big)\Omega\>=\sum_{\zeta,\xi\in K}\sqrt{t_\zeta t_\xi}\<\zeta|\Phi(B)\xi\>\<\zeta|A'\xi\>\nonumber\\
&=&\sum_{\zeta,\xi\in K}\sqrt{t_\zeta t_\xi}\<\zeta|\Phi(B)\xi\>\<\xi|(A')^T\zeta\>=\tr{\rho_0^{1/2}(A')^T\rho_0^{1/2}\Phi(B)}.\label{eq:recovery1}
\end{eqnarray}
Especially, we find
\begin{equation}\label{eq:recovery2}
\<\zeta|\Phi(B)\xi\>=\frac{1}{\sqrt{t_\zeta t_\xi}}\tr{S_\Omega^\Phi(|\zeta\>\<\xi|\otimes B)},\qquad\zeta,\,\xi\in K,\quad B\in\mc L(\mc K).
\end{equation}

Using the definitions of the modular structures, it is simple to check that, in this type-I case, $\Delta=\rho_0\otimes\rho_0^{-1}$ with a suitably defined domain and $j(A\otimes A')=\ovl{A'}\otimes\ovl{A}$ for all $A,\,A'\in\mc L(\hil)$, where $\ovl A=(A^T)^*=(A^*)^T=:A^{T\,*}$, i.e.,\ $\<\zeta|\ovl A\xi\>=\ovl{\<\zeta|A\xi\>}$, $\zeta,\,\xi\in K$, for all $A\in\mc L(\hil)$. Hence, the map $A\mapsto j(A\otimes\id_\hil)^*$ is truly essentially the transpose map defined by the basis fixed when choosing the GNS-construction for $\rho_0$.

We move on to establish a connection between Kraus decompositions of quantum channels and spectral decompositions of their Choi states. First, however, we have to establish some basics in the dilation theory of channels and its relation to Kraus decompositions. Let $\Phi\in{\bf NCH}(\hil,\mc K)$ be a normal channel. We say that a pair $(\mc L,V)$ consisting of a Hilbert space $\mc L$ and an isometry $V:\hil\to\mc K\otimes\mc L$ constitutes a {\it Stinespring dilation} for $\Phi$ if $\Phi(B)=V^*(B\otimes\id_{\mc L})V$ for all $B\in\mc L(\mc K)$ or, equivalently. If, additionally, the vectors $(B\otimes\id_{\mc L})V\fii$, $B\in\mc L(\mc K)$, $\fii\in\hil$, span a dense subspace of $\mc K\otimes\mc L$, the dilation is {\it minimal}. This definition of a (minimal) Stinespring dilation coincides with the one presented in the beginning of this paper for channels $\Phi:\mc A\to\mc L(\hil)$ in the case where $\mc A=\mc L(\mc K)$ and $\Phi$ is normal. Every normal channel $\Phi\in{\bf NCH}(\hil,\mc K)$ has a minimal Stinespring dilation $(\mc L_0,V_0)$ and, for any other dilation $(\mc L,V)$ for $\Phi$, there is an isometry $W:\mc L_0\to\mc L$ such that $V=(\id_{\mc K}\otimes W)V_0$. Particularly if $(\mc L_0,V_0)$ and $(\mc L_1,V_1)$ are both minimal Stinespring dilations for $\Phi$, there is a unique unitary $U:\mc L_0\to\mc L_1$ such that $V_1=(\id_{\mc K}\otimes U)V_0$.

Operators $K_\lambda:\hil\to\mc K$, $\lambda\in L$, where $L\neq\emptyset$, are said to constitute a {\it Kraus decomposition} or are {\it Kraus operators} for $\Phi\in{\bf NCH}(\hil,\mc K)$ if $\Phi(B)=\sum_{\lambda\in L}K_\lambda^*BK_\lambda$ for all $B\in\mc L(\mc K)$ where the series converges weakly. There is a connection between Stinespring dilations and Kraus decompositions: Any Stinespring dilation $(\mc L,V)$ for $\Phi$ and any orthonormal basis $\{\eta_\lambda\}_{\lambda\in L}\subset\mc L$ defines a set $\{K_\lambda\}_{\lambda\in L}$ of Kraus operators for $\Phi$ by defining the linear operators $V_\lambda:\mc K\to\mc K\otimes\mc L$, $V_\lambda\psi=\psi\otimes\eta_\lambda$, $\psi\in\mc K$, $\lambda\in L$, and setting $K_\lambda=V_\lambda^*V$ for all $\lambda\in L$. In this situation, we say that $\{K_\lambda\}_{\lambda\in L}$ {\it arise from the dilation $(\mc L,V)$ and the orthonormal basis $\{\eta_\lambda\}_{\lambda\in L}\subset\mc L$}. Moreover, any set $\{K_\lambda\}_{\lambda\in L}$, $L\neq\emptyset$, of Kraus operators for $\Phi$ arises from a Stinespring dilation and a choice of an orthonormal basis of the dilation space. Indeed, let $\{\eta_\lambda\}_{\lambda\in L}$ be the natural basis for $\ell^2_L$ and define $V:\hil\to\mc K\otimes\ell^2_L$, $V\fii=\sum_{\lambda\in L}K_\lambda\fii\otimes\eta_\lambda$, $\fii\in\hil$. It follows that $(\ell^2_L,V)$ is a Stinespring dilation for $\Phi$ and the original Kraus decomposition can be recovered in the same way as above.

Let $\{K_\lambda\}_{\lambda\in L}$ be a Kraus decomposition for $\Phi\in{\bf NCH}(\hil,\mc K)$. Since, for any $\fii\in\hil$ and $\psi\in\mc K$,
$$
\sum_{\lambda\in L}|\<\psi|K_\lambda\fii\>|^2=\sum_{\lambda\in L}\<K_\lambda\fii|\psi\>\<\psi|K_\lambda\fii\>\leq\|\psi\|^2\sum_{\lambda\in L}\<\fii|K_\lambda^*K_\lambda\fii\>=\|\psi\|^2\|\fii\|^2<\infty,
$$
it follows that, for any square-summable sequence $(\alpha_\lambda)_{\lambda\in L}\subset\C$ and $\fii\in\hil$ and $\psi\in\mc K$, $\big|\sum_{\lambda\in L}\alpha_\lambda\<\psi|K_\lambda\fii\>\big|<\infty$, according to the Cauchy-Schwarz inequality. Let us make a useful definition.

\begin{definition}
Let $\hil$ and $\mc K$ be Hilbert spaces. Linear operators $K_\lambda:\hil\to\mc K$, $\lambda\in L$, $L\neq\emptyset$, such that $\sum_{\lambda\in L}\|K_\lambda\fii\|^2<\infty$ for all $\fii\in\hil$, are {\it weakly independent} if, for a net $(\alpha_\lambda)_{\lambda\in L}$ of complex numbers such that $\sum_{\lambda\in L}|\alpha_\lambda|^2<\infty$, $\sum_{\lambda\in L}\alpha_\lambda\<\psi|K_\lambda\fii\>=0$ for all $\fii\in\hil$ and $\psi\in\mc K$, only if $\alpha_\lambda=0$ for all $\lambda\in L$.
\end{definition}

\begin{lemma}\label{lemma:StinespringKraus}
Let $\hil$ and $\mc K$ be Hilbert spaces and $\Phi\in{\bf NCH}(\hil,\mc K)$. A set $\{K_\lambda\}_{\lambda\in L}$, $L\neq\emptyset$, of Kraus operators for $\Phi$ is weakly independent if and only if it arises from a minimal Stinespring dilation $(\mc L,V)$ for $\Phi$ and some orthonormal basis $\{\eta_\lambda\}_{\lambda\in L}\subset\mc L$ in the way defined above.
\end{lemma}

\begin{proof}
Assume first that the Kraus operators $K_\lambda$, $\lambda\in L$, are weakly independent. Let $\{\eta_\lambda\}_{\lambda\in L}$ be the natural basis for the square-summable sequence space $\ell^2_L$ and define $V:\hil\to\mc K\otimes\ell^2_L$, $V\fii=\sum_{\lambda\in L}K_\lambda\fii\otimes\eta_\lambda$ so that $(\ell^2_L,V)$ is a Stinespring dilation for $\Phi$. Let us prove that this dilation is also minimal. Let $v=\sum_{\lambda\in L}\psi_\lambda\otimes\eta_\lambda$ be a general vector in $\mc K\otimes\ell^2_L$, i.e.,\ $\sum_{\lambda\in L}\|\psi_\lambda\|^2<\infty$. Pick $\fii\in\hil$, $\psi\in\mc K\setminus\{0\}$, and $B_0\in\mc L(\mc K)$. Define $V_\psi:\ell^2_L\to\mc K\otimes\ell^2_L$, $V_\psi\tj=\psi\otimes\tj$ for all $\tj\in\ell^2_L$. It follows that $\sum_{\lambda\in L}|\<\psi_\lambda|\psi\>|^2=\|V_\psi^*v\|^2<\infty$. Denote $B=|\psi\>\<\psi|B_0$. Assume that
$$
0=\<v|(B\otimes\id_{\ell^2_L})V\fii\>=\sum_{\lambda\in L}\<\psi_\lambda|\psi\>\<B_0^*\psi|K_\lambda\fii\>
$$
for all $\fii\in\hil$, $\psi\in\mc K\setminus\{0\}$, and $B_0\in\mc L(\mc K)$. Varying $\fii\in\hil$ and $B_0\in\mc L(\mc K)$, we see that $\sum_{\lambda\in L}\<\psi_\lambda|\psi\>\<\psi'|K_\lambda\fii'\>=0$ for all $\fii'\in\hil$ and $\psi'\in\mc K$. According to the weak independence of the Kraus operators, this means that $\<\psi_\lambda|\psi\>=0$ for all $\lambda\in L$. Since this holds for all $\psi\in\mc K\setminus\{0\}$, we have that $\psi_\lambda=0$ for all $\lambda\in L$, i.e.,\ $v=0$. This shows the minimality of $(\ell^2_L,V)$.

Let now $(\mc L,V)$ be a minimal Stinespring dilation for $\Phi$, $\{\eta_\lambda\}_{\lambda\in L}$ be an orthonormal basis of $\mc L$, and $\{K_\lambda\}_{\lambda\in L}$ be the set of Kraus operators for $\Phi$ arising from $(\mc L,V)$ and $\{\eta_\lambda\}_{\lambda\in L}$. Let us show that the set $\{K_\lambda\}_{\lambda\in L}$ is weakly independent. Let $(\alpha_\lambda)_{\lambda\in L}$ be a net of complex numbers such that $\sum_{\lambda\in L}|\alpha_\lambda|^2<\infty$ and define $\tj:=\sum_{\lambda\in L}\ovl{\alpha_\lambda}\eta_\lambda\in\mc L$. Suppose that
$$
0=\sum_{\lambda\in L}\alpha_\lambda\<\psi|K_\lambda\fii\>=\sum_{\lambda\in L}\alpha_\lambda\<\psi\otimes\eta_\lambda|V\fii\>=\<\psi\otimes\tj|V\fii\>
$$
for all $\fii\in\hil$ and $\psi\in\mc K$. Denoting $\psi=B^*\psi_0$ for some $\psi_0\in\mc K\setminus\{0\}$ and varying $B\in\mc L(\mc K)$, we especially have that $\<\psi_0\otimes\tj|(B\otimes\id_{\mc L})V\fii\>=0$ for all $B\in\mc L(\mc K)$ and $\fii\in\hil$. The minimality of $(\mc L,V)$ implies that $\psi_0\otimes\tj=0$ and, since $\psi\neq0$, $\tj=0$, i.e.,\ $\alpha_\lambda=0$ for all $\lambda\in L$.
\end{proof}

The preceding lemma justifies the following definition:

\begin{definition}
Let $\hil$ and $\mc K$ be Hilbert spaces and $\Phi\in{\bf NCH}(\hil,\mc K)$. We say that $\{K_\lambda\}_{\lambda\in L}$, $L\neq\emptyset$, where $K_\lambda:\hil\to\mc K$ are linear operators for all $\lambda\in L$, is a {\it minimal set of Kraus operators} for $\Phi$ if they provide a Kraus decomposition for $\Phi$ and are weakly independent.
\end{definition}

For any $\rho\in\mc S(\hil)$ and $\Phi\in{\bf NCH}(\hil,\mc K)$, there is a (minimal) set $\{K_\lambda\}_{\lambda\in L}$ of Kraus operators for $\Phi$ such that, whenever $\lambda\neq\lambda'$, $\tr{\rho K_\lambda^*K_{\lambda}}=0$. Indeed, for a (minimal) Stinespring dilation $(\mc L,V)$ for $\Phi$, let $\{\eta_\lambda\}_{\lambda\in L}$ be an orthonormal basis diagonalizing the state ${\rm tr}_\hil[V\rho V^*]$ and $\{K_\lambda\}_{\lambda\in L}$ be the (minimal) set of Kraus operators arising from $(\mc L,V)$ and $\{\eta_\lambda\}_{\lambda\in L}$. Whenever $\lambda,\,\lambda'\in L$, $\lambda\neq\lambda'$,
$$
\tr{\rho K_\lambda^*K_{\lambda'}}=\tr{V\rho V^*(\id_\hil\otimes|\eta_\lambda\>\<\eta_{\lambda'}|)}=\<\eta_{\lambda'}|{\rm tr}_\hil [V\rho V^*]\eta_\lambda\>=0.
$$
It is easy to see that, whenever $\{K_\lambda\}_{\lambda\in L}$ is a set of (non-zero) Kraus operators for a channel $\Phi$ such that $\tr{\rho_0 K_\lambda^*K_{\lambda'}}=0$ whenever $\lambda\neq\lambda'$ for some faithful $\rho_0\in\mc S(\hil)$, $\{K_\lambda\}_{\lambda\in L}$ is a minimal set of Kraus operators for $\Phi$; see the end of the proof of the following proposition for this.

\begin{proposition}\label{prop:KrausSp}
Let $\hil$ be a separable Hilbert space and $\mc K$ be another Hilbert space. Pick a faithful state $\rho_0\in\mc S(\hil)$ and a minimal GNS-vector $\Omega$ for $\rho_0$.
\begin{itemize}
\item[(a)] Let $\Phi\in{\bf NCH}(\hil,\mc K)$ and pick a set $\{K_\lambda\}_{\lambda\in L}$, $L\neq\emptyset$, of Kraus operators for $\Phi$ such that $\tr{\rho_0K_\lambda^*K_{\lambda'}}=0$ whenever $\lambda\neq\lambda'$. The set $\{w_\lambda\}_{\lambda\in L}\subset\hil\otimes\mc K$, $w_\lambda=(\id_\hil\otimes K_\lambda)\Omega$, $\lambda\in L$, is orthogonal and $S_\Omega^\Phi=\sum_{\lambda\in L}|w_\lambda\>\<w_\lambda|$.
\item[(b)] Suppose that $S\in\mc S(\hil\otimes\mc K)$ is such that ${\rm tr}_{\mc K}[S]=\rho_0$ and let $\Phi\in{\bf NCH}(\hil,\mc K)$ be the channel such that $S=S_\Omega^\Phi$. For any orthogonal decomposition $S=\sum_{\lambda\in L}|w_\lambda\>\<w_\lambda|$ such that $w_\lambda\in\hil\otimes\mc K\setminus\{0\}$ for all $\lambda\in L$, $L\neq\emptyset$, there is a minimal set $\{K_\lambda\}_{\lambda\in L}$ of Kraus operators for $\Phi$ such that $\tr{\rho_0K_\lambda^*K_{\lambda'}}=0$ whenever $\lambda\neq\lambda'$ and $w_\lambda=(\id_\hil\otimes K_\lambda)\Omega$ for all $\lambda\in L$.
\end{itemize}
\end{proposition}

\begin{proof}
Throughout this proof, we may fix a spectral decomposition $\rho_0=\sum_{\xi\in K}t_\xi|\xi\>\<\xi|$ where $K\subset\hil$ is an orthogonal basis and $t_\xi>0$ for all $\xi\in K$ are such that $\sum_{\xi\in K}t_\xi=1$ and we may choose $\Omega=\sum_{\xi\in K}\sqrt{t_\xi}\xi\otimes\xi$. This choice does not restrict the generality of this proof and makes the calculations more straightforward.

Let us prove item (a). Define $w_\lambda$, $\lambda\in L$, as in the claim. It easily follows that, whenever $\lambda\neq\lambda'$,
$$
\<w_\lambda|w_{\lambda'}\>=\<\Omega|(\id_\hil\otimes K_\lambda^*K_{\lambda'})\Omega\>=\tr{\rho_0K_\lambda^*K_{\lambda'}}=0,
$$
and
$$
S_\Omega^\Phi=({\rm id}\otimes\Phi)(|\Omega\>\<\Omega|)=\sum_{\lambda\in L}|(\id_\hil\otimes K_\lambda)\Omega\>\<(\id_\hil\otimes K_\lambda)\Omega|=\sum_{\lambda\in L}|w_\lambda\>\<w_\lambda|.
$$

Let us go on to proving item (b). Let $S=\sum_{\lambda\in L}|w_\lambda\>\<w_\lambda|$ be an orthogonal decomposition with non-zero vectors $w_\lambda\in\hil\otimes\mc K$, $\lambda\in L$, where $L$ is some non-empty set. For each $\lambda\in L$, let $\{\psi_{\lambda,\xi}\}_{\xi\in K}\subset\mc K$ be a set such that $\sum_{\xi\in K}\|\psi_{\lambda,\xi}\|^2<\infty$ and $w_\lambda=\sum_{\xi\in K}\xi\otimes\psi_{\lambda,\xi}$. Define the linear operator $V$ on the linear span of the basis $K$ and with the target space $\hil\otimes\ell^2_L$ through $V\xi=t_\xi^{-1/2}\sum_{\lambda\in L}\psi_{\lambda,\xi}\otimes\eta_\lambda$ where $\{\eta_\lambda\}_{\lambda\in L}$ is the natural basis of $\ell^2_L$. We have, for all $\zeta,\,\xi\in K$,
\begin{eqnarray*}
\<V\zeta|V\xi\>&=&\frac{1}{\sqrt{t_\zeta t_\xi}}\sum_{\lambda\in L}\<\psi_{\lambda,\zeta}|\psi_{\lambda,\xi}\>=\frac{1}{\sqrt{t_\zeta t_\xi}}\sum_{\lambda\in L}\<w_\lambda|(|\zeta\>\<\xi|\otimes\id_{\mc K})w_\lambda\>\\
&=&\frac{1}{\sqrt{t_\zeta t_\xi}}\tr{S(|\zeta\>\<\xi|\otimes\id_{\mc K})}=\frac{1}{\sqrt{t_\zeta t_\xi}}\<\xi|\rho_0\zeta\>=\left\{\begin{array}{ll}
1&{\rm when}\ \zeta=\xi,\\
0&{\rm otherwise},
\end{array}\right.
\end{eqnarray*}
implying that $V$ can be extended into an isometry $V:\hil\to\mc K\otimes\ell^2_L$. Moreover, for all $\zeta,\,\xi\in K$ and $B\in\mc L(\mc K)$, we have
\begin{eqnarray*}
\<\zeta|\Phi(B)\xi\>&=&\frac{1}{\sqrt{t_\zeta t_\xi}}\tr{S(|\zeta\>\<\xi|\otimes B)}=\frac{1}{\sqrt{t_\zeta t_\xi}}\<w_\lambda|(|\zeta\>\<\xi|\otimes B)w_\lambda\>\\
&=&\frac{1}{\sqrt{t_\zeta t_\xi}}\sum_{\lambda\in L}\sum_{\zeta',\xi'\in K}\<\zeta'\otimes\psi_{\lambda,\zeta'}|(|\zeta\>\<\xi|\otimes B)(\xi'\otimes\psi_{\lambda,\xi'})\>\\
&=&\frac{1}{\sqrt{t_\zeta t_\xi}}\sum_{\lambda\in L}\<\psi_{\lambda,\zeta}|B\psi_{\lambda,\xi}\>=\<V\zeta|(B\otimes\id_{\ell^2_L})V\xi\>,
\end{eqnarray*}
implying that $(\ell^2_L,V)$ is a Stinespring dilation for $\Phi$. Let $\{K_\lambda\}_{\lambda\in L}$ be the set of Kraus operators for $\Phi$ arising from $(\ell^2_L,V)$ and $\{\eta_\lambda\}_{\lambda\in L}$. It follows that $K_\lambda\xi=t_\xi^{-1/2}\psi_{\lambda,\xi}$ for all $\lambda\in L$ and $\xi\in K$. We find 
$$
(\id_\hil\otimes K_\lambda)\Omega=\sum_{\xi\in K}\sqrt{t_\xi}\xi\otimes K_\lambda\xi=\sum_{\xi\in K}\xi\otimes\psi_{\lambda,\xi}=w_\lambda
$$
for all $\lambda\in L$. Moreover, for any $\lambda,\,\lambda'\in L$, $\lambda\neq\lambda'$,
$$
\tr{\rho_0K_\lambda^*K_{\lambda'}}=\sum_{\xi\in K}t_\xi\<K_\lambda\xi|K_{\lambda'}\xi\>=\sum_{\xi\in K}\<\psi_{\lambda,\xi}|\psi_{\lambda',\xi}\>=\<w_\lambda|w_{\lambda'}\>=0.
$$

Assume that $\alpha_\lambda\in\C$, $\lambda\in L$, are such that $\sum_{\lambda\in L}|\alpha_\lambda|^2<\infty$ and $\sum_{\lambda\in L}\alpha_\lambda\<\psi|K_\lambda\fii\>=0$ for all $\fii\in\hil$ and $\psi\in\mc K$. It follows that $\sum_{\lambda,\lambda'\in L}\ovl{\alpha_\lambda}\alpha_{\lambda'}\<K_\lambda\fii|K_{\lambda'}\fii\>=0$ for all $\fii\in\hil$. From this it follows that
$$
0=\sum_{\lambda,\lambda'\in L}\ovl{\alpha_\lambda}\alpha_{\lambda'}\tr{\rho_0K_\lambda^*K_{\lambda'}}=\sum_{\lambda\in L}|\alpha_\lambda|^2\|w_\lambda\|^2,
$$
implying, since $w_\lambda\neq0$ for all $\lambda\in L$, that $\alpha_\lambda=0$ for all $\lambda\in L$. Thus $\{K_\lambda\}_{\lambda\in L}$ is a minimal set of Kraus operators for $\Phi$.
\end{proof}

\subsection{Transposed channels and their Choi-Jamio\l kowski states}\label{subsec:transpose}

The transposed channels play a role in sufficiency questions and quantum information retrieval \cite{JencovaPetz2006}. We will see that a channel and its transpose share essentially the same Choi-Jamio\l kowski state.

Let $\mf M$ and $\mf N$ be $\sigma$-finite von Neumann algebras, $\rho_0\in{\bf NS}_1(\mf M)$ and $\rho_1\in{\bf NS}_1(\mf N)$ be faithful, $(\mf H_i,\Omega_i)$ be a GNS-construction for $\rho_i$, $i=0,\,1$, where $\Omega_0$ is cyclic and separating for $(\mf H_0,\mf M)$ and $\Omega_1$ is cyclic and separating for $(\mf H_1,\mf N)$, let $j_{\mf M}$ be the modular conjugation associated to $\Omega_0$ and $j_{\mf N}$ be the modular conjugation associated to $\Omega_1$, and denote by $\mf M'$ the commutant of $\mf M$ within $\mc L(\mf H_0)$ and by $\mf N'$ the commutant of $\mf N$ within $\mc L(\mf H_1)$. The following definition slightly modifies the well-known concept of transposed channels. The map $\Phi_{\Omega_0,\Omega_1}^\#$ in the definition below is well defined, and in the case where $\rho_1=\rho_0\circ\Phi$, one can directly consult \cite[Proposition 3.1]{AcCe82} on the matter. In the more general case, the proof simply uses standard methods of dilation theory and is similar to the latter half of the proof of Theorem \ref{theor:CJisomorphism}.

\begin{definition}\label{def:transposed}
Let us make the above assumptions on the von Neumann algebras $\mf M$ and $\mf N$. Whenever $\Phi:\mf N\to\mf M$ is a normal (completely) positive linear map such that there is $\lambda\geq0$ such that $\rho_0\circ\Phi\leq\lambda\rho_1$, the unique normal (completely) positive linear map $\Phi_{\Omega_0,\Omega_1}^\#:\mf M'\to\mf N'$ defined through
$$
\<\Omega_1|\Phi_{\Omega_0,\Omega_1}^\#(a')b\Omega_1\>=\<\Omega_0|a'\Phi(b)\Omega_0\>,\qquad a'\in\mf M',\quad b\in\mf N,
$$
is called as the {\it $(\Omega_0,\Omega_1)$-commutant dual of $\Phi$}, and $\Phi_{\Omega_0,\Omega_1}^T:=j_{\mf N}\circ\Phi_{\Omega_0,\Omega_1}^\#\circ j_{\mf M}:\mf M\to\mf N$ is called as the {\it $(\Omega_0,\Omega_1)$-transpose of $\Phi$}. When $\rho_1=\rho_0\circ\Phi$, we denote $\Phi_{\Omega_0,\Omega_1}^\#=:\Phi_{\Omega_0}^\#$ and $\Phi_{\Omega_0,\Omega_1}^T=:\Phi_{\Omega_0}^T$.
\end{definition}

Let us additionally assume that $\mf L$ is a $\sigma$-finite von Neumann algebra with a faithful state $\rho_2$ and that $\rho_2$ has a GNS-construction $(\mf H_2,\Omega_2)$ where $\Omega_2$ is cyclic and separating for $(\mf H_2,\mf L)$. Assume that $\Psi:\mf L\to\mf N$ and $\Phi:\mf N\to\mf M$ are normal (completely) positive linear maps such that $\rho_1\circ\Psi\leq\lambda_1\rho_2$ and $\rho_0\circ\Phi\leq\lambda_0\rho_1$ for some $\lambda_0,\,\lambda_1\leq0$. Then one easily finds that
$$
(\Phi\circ\Psi)_{\Omega_0,\Omega_2}^\#=\Psi_{\Omega_1,\Omega_2}^\#\circ\Phi_{\Omega_0,\Omega_1}^\#,\quad (\Phi\circ\Psi)_{\Omega_0,\Omega_2}^T=\Psi_{\Omega_1,\Omega_2}^T\circ\Phi_{\Omega_0,\Omega_1}^T.
$$

Let $\mf M$ and $\mf N$ both be von Neumann algebras where $\mf M$ is $\sigma$-finite and possesses a faithful state $\rho_0\in{\bf NS}_1(\mf M)$. Whenever $\Phi\in{\bf NCH}(\mf M,\mf N)$, we may, without any essential loss of generality, assume that $\rho_0\circ\Phi\in{\bf NS}_1(\mf N)$ is faithful as well. To see this, suppose that $p\in\mf N$ is the support of $\Phi$ \cite[Section 10.8]{kirja}, i.e.,\ the infimum of the projections $q\in\mf N$ such that $\Phi(q)=1_{\mf M}$. Naturally, $p$ is a projection as well. It follows that $\Phi(b)=\Phi(pbp)$ for all $b\in\mf N$. Define the von Neumann algebra $\mf N_p:=p\mf Np$; naturally we may restrict states $\sigma\in{\bf NS}_1(\mf N)$ onto $\mf N_p$. Define $\Phi_p\in{\bf NCH}(\mf M,\mf N_p)$, $\Phi_p=\Phi|_{\mf N_p}$. Define $\kappa\in{\bf NCH}(\mf N_p,\mf N)$, $\kappa(b)=pbp$ for all $b\in\mf N$. It follows that, for all $b\in\mf N$, $(\Phi_p\circ\kappa)(b)=\Phi(pbp)=\Phi(b)$, i.e.,\ $\Phi_p\circ\kappa=\Phi$. Similarly, when we define $\iota:\mf N_p\to\mf N$ as the natural inclusion, $\Phi\circ\iota=\Phi_p$. Suppose now that $\tilde{b}\in\mf N_p$, $\tilde{b}\geq0$, and $(\rho_0\circ\Phi_p)(\tilde{b})=0$. Since $\rho_0$ is faithful, $\Phi_p(\tilde{b})=0$, i.e.,\ $0=\Phi(\tilde{b})=\Phi(p\tilde{b}p)$, implying that $\tilde{b}=p\tilde{b}p=0$. This means that we may always restrict $\Phi$ onto an algebra such that the restriction is, in a sense, equivalent with $\Phi$ and $\rho_0\circ\Phi$ is faithful on this algebra. Thus, the assumption made above on the $\sigma$-finiteness of $\mf N$ is actually redundant, and we may always assume that $\rho_0\circ\Phi$ is faithful.

\begin{proposition}\label{prop:transposed}
Let $\mf M$ and $\mf N$ be von Neumann algebras where $\mf M$ is injective and $\sigma$-finite and possesses a faithful state $\rho_0$. Let $\Phi\in{\bf NCH}(\mf M,\mf N)$ and assume that $\rho_1:=\rho_0\circ\Phi$ is faithful. We equip $\rho_i$ with a GNS-construction $(\mf H_i,\Omega_i)$ where $\Omega_i$ is cyclic and separating, $i=1,\,2$. We have, for all $a\in\mf M$, $a'\in\mf M'$, $b\in\mf N$, and $b'\in\mf N'$,
\begin{eqnarray*}
S_{\Omega_1}^{\Phi_{\Omega_0}^\#}(b\otimes a')&=&S_{\Omega_0}^\Phi(a'\otimes b),\\
S_{\Omega_1}^{\Phi_{\Omega_0}^T}(b'\otimes a)&=&\ovl{S_{\Omega_0}^\Phi\big(j_{\mf M}(a)\otimes j_{\mf N}(b')\big)}.
\end{eqnarray*}
\end{proposition}

\begin{proof}
We have, for all $a'\in\mf M'$ and $b\in\mf N$,
$$
S_{\Omega_1}^{\Phi_{\Omega_0}^\#}(b\otimes a')=\<\Omega_1|b\Phi_{\Omega_0}^\#(a')\Omega_1\>=\<\Omega_0|a'\Phi(b)\Omega_0\>=S_{\Omega_0}^\Phi(a'\otimes b).
$$
Moreover, for all $a\in\mf M$ and $b'\in\mf N'$,
\begin{eqnarray*}
S_{\Omega_1}^{\Phi_{\Omega_0}^T}(b'\otimes a)&=&\<\Omega_1|b'\Phi_{\Omega_0}^T(a)\Omega_1\>=\<\Omega_1|b'J_{\mf N}(\Phi_{\Omega_0}^\#\circ j_{\mf M})(a)\Omega_1\>\\
&=&\<j_{\mf N}(b')(\Phi_{\Omega_0}^\#\circ j_{\mf M})(a)\Omega_1|\Omega_1\>=\ovl{\<\Omega_1|(\Phi_{\Omega_0}^\#\circ j_{\mf M})(a)j_{\mf N}(b')\Omega_1\>}\\
&=&\ovl{\<\Omega_0|j_{\mf M}(a)(\Phi\circ j_{\mf N})(b')\Omega_0\>}=\ovl{S_{\Omega_0}^\Phi\big(j_{\mf M}(a)\otimes j_{\mf N}(b')\big)}.
\end{eqnarray*}
\end{proof}

When $\mf M=\mc L(\hil)$ and $\mf N=\mc L(\mc K)$ with some separable Hilbert spaces $\hil$ and $\mc K$, $\Phi\in{\bf NCH}(\hil,\mc K)$, $\rho_0\in\mc S(\hil)$ is faithful, and we assume (without any essential loss of generality) that $\rho_1:=\Phi_*(\rho_0)$ is faithful as well, $\Phi$ and its $\rho_0$-transpose are associated through
\begin{equation}\label{eq:transposed}
\rho_1^{1/2}\Phi_{\rho_0}^T(A)\rho_1^{1/2}=\Phi_*(\rho_0^{1/2}A\rho_0^{1/2}),\qquad A\in\mc L(\hil);
\end{equation}
We implicitly choose GNS-vectors arising from spectral decompositions of $\rho_0$ and $\rho_1$ (with real phases) and only refer to the state $\rho_0$ in the transpose. This can be seen as follows: When we identify $\mc L(\hil)'=\mc L(\hil)$ and $\mc L(\mc K)'=\mc L(\mc K)$, the involution $j_{\mc L(\hil)}$ when restricted to $\mc L(\hil)$ and defined to take values in $\mc L(\hil)'=\mc L(\hil)$ has already seen to be the matrix-entry-wise complex conjugation with respect to the basis $K$ where $\rho_0$ is diagonalized, i.e.,\ $j_{\mc L(\hil)}(A)=A^{T\,*}$ where the transpose is defined with respect to the same basis $K$. The same applies, naturally to $j_{\mc L(\mc K)}$. Using Equation \eqref{eq:recovery1} and Proposition \ref{prop:transposed}, we now have, for all $A\in\mc L(\hil)$ and $B'\in\mc L(\mc K)$,
\begin{eqnarray*}
\tr{\rho_1^{1/2}(B')^T\rho_1^{1/2}\Phi_{\rho_0}^T(A)}&=&\tr{S_{\Omega_1}^{\Phi_{\rho_0}^T}(B'\otimes A)}=\ovl{\tr{S_{\Omega_0}^\Phi\big(A^{T\,*}\otimes(B')^{T\,*}\big)}}\\
&=&\ovl{\tr{\rho_0^{1/2}A^*\rho_0^{1/2}\Phi\big((B')^T\big)^*}}=\tr{\Phi\big((B')^T\big)\rho_0^{1/2}A\rho_0^{1/2}},
\end{eqnarray*}
from where Equation \eqref{eq:transposed} now follows. Clearly, $\Phi_{\rho_0}^\#(A)=\Phi_{\rho_0}^T(A^T)^{T'}$ for all $A\in\mc L(\hil)$ where $\mc L(\mc K)\ni B\mapsto B^{T'}\in\mc L(\mc K)$ is the transpose defined with respect to the chosen eigenbasis of $\rho_1$.

\section{Choi-Jamio\l kowski states of covariant channels}\label{sec:covCJ}

We now go on to applying the above established Choi-Jamio\l kowski isomorphism to covariant channels. We focus on the case of normal channels. After introducing some basic definitions and specifying what we mean with `covariant channels' (the definition used here is the usual one), we establish necessary and sufficient conditions for the covariance of a channel using its Choi state. In the following subsection, we concentrate again on the case of type-I factors.

Let us first recall some terminology. Pick a group $G$ and a Hilbert space $\hil$. We denote the group of unitary operators on $\hil$ by $\mc U(\hil)$. A {\it unitary representation of $G$ on $\hil$} is a homomorphism $U:G\to\mc U(\hil)$. A {\it projective unitary representation of $G$ on $\hil$} is a map $U:G\to\mc U(\hil)$ such that $U(e)=\id_\hil$, $e$ being the neutral element of $G$, and $U(gh)=m(g,h)U(g)U(h)$ for all $g,\,h\in G$, where $m:G\times G\to\T$, $\T=\{z\in\C\,|\,|z|=1\}$, is a {\it multiplier}, i.e., $m(e,g)=m(g,e)=1$ for all $g\in G$ and $m(g,h)m(gh,k)=m(g,hk)m(h,k)$ for all $g,\,h,\,k\in G$. A map $g\mapsto\alpha_g$ defined on $G$ and taking values in the group of automorphisms of $\mc L(\hil)$ is a group homomorphism if and only if there is a projective unitary representation $U:G\to\mc U(\hil)$ such that $\alpha_g(A)=U(g)AU(g)^*$ for all $g\in G$ and $A\in\mc L(\hil)$ \cite[Theorem 7.5]{Varadarajan}.

In what follows $\mf M$ is an injective $\sigma$-finite von Neumann algebra and we fix a faithful state $\rho_0\in{\bf NS}_1(\mf M)$ together with a GNS construction $(\mf H,\Omega)$ for $\rho_0$ where $\Omega$ is a cyclic and separating vector for $(\mf H,\mf M)$ and let $S$, $\Delta$, $J$, and $j$ be the associated modular structures. Denote by $V_\Omega$ the strong closure of the set $\{aj(a)\Omega\,|\,a\in\mf M\}$ and by $V_\Omega^1$ the section of norm-1 vectors in $V_\Omega$. According to Araki in \cite{Araki74}, there is a bijective homeomorphism $\chi:{\bf NS}_1(\mf M)\to V_\Omega^1$ such that $\<\chi(\rho)|a\chi(\rho)\>=\rho(a)$ for all $\rho\in{\bf NS}_1(\mf M)$ and $a\in\mf M$.

We denote the group of normal automorphisms of $\mf M$ by ${\rm Aut}(\mf M)$. We equip ${\rm Aut}(\mf M)$ with point-predual-norm topology, i.e.,\ the coarsest topology with respect to which the maps ${\rm Aut}(\mf M)\ni\alpha\mapsto\rho\circ\alpha\in\mf M_*$, $\rho\in\mf M_*$, are continuous with respect to the Banach space topology of $\mf M_*$. According to \cite[Theorem 11 and subsequent discussion]{Araki74}, there is a strongly continuous unitary representation ${\rm Aut}(\mf M)\ni\alpha\mapsto U_\alpha\in\mc U(\mf H)$ such that $\alpha(a)=U_\alpha aU_\alpha^*$ for all $\alpha\in{\rm Aut}(\mf M)$ and $a\in\mf M$. Note that injectivity is not required for the existence of this representation. Moreover, $U_\alpha J=JU_\alpha$ and $U_\alpha^*\chi(\rho)=\chi(\rho\circ\alpha)$ for all $\alpha\in{\rm Aut}(\mf M)$ and $\rho\in{\bf NS}_1(\mf M)$. Denote by $\mf M'$ the commutant of $\mf M$ within $\mc L(\mf H)$. There is an isomorphism ${\rm Aut}(\mf M)\ni\alpha\mapsto\alpha'\in{\rm Aut}(\mf M')$ given by $\alpha'=j\circ\alpha\circ j$ for all $\alpha\in{\rm Aut}(\mf M)$. Using the fact that $\alpha\mapsto U_\alpha$ commutes with $J$, we find $\alpha'(a')=JU_\alpha Ja'JU_\alpha^*J=U_\alpha a'U_\alpha^*$.

For any $\rho\in{\bf NS}_1(\mf M)$, denote the group of those $\alpha\in{\rm Aut}(\mf M)$ such that $\rho\circ\alpha=\rho$ by $A_\rho$. As a preimage of the compact set $\{0\}$ in the continuous map ${\rm Aut}(\mf M)\ni\alpha\mapsto\rho-\rho\circ\alpha\in\mf M_*$, $A_\rho$ is closed. We have $U_\alpha^*\chi(\rho)=\chi(\rho\circ\alpha)=\chi(\rho)$ for all $\alpha\in A_\rho$ and $\rho\in{\bf NS}_1(\mf M)$. Especially, $U_{\alpha}^*\Omega=\Omega$ for all $\alpha\in A_{\rho_0}$.

When $G$ is a group and $\mf M$ is a von Neumann algebra, we say that a map $G\ni g\mapsto\alpha_g\in{\rm Aut}(\mf M)$ is a {\it $G$-action on $\mf M$} if it is a group homomorphism. If, additionally, $G$ is a topological group, and the action $g\mapsto\alpha_g$ is continuous with respect to the group topology and the point-predual-norm topology of ${\rm Aut}(\mf M)$, we simply say that the action is continuous.

Throughout this section, we also assume that $\mf N$ is a von Neumann algebra. In many of the cases, we could just assume that $\mf N$ is only a unital $C^*$-algebra and use the non-normal Choi-Jamio\l kowski isomorphism but, for simplicity and physical motivation, we only consider the full von Neumann algebra case and binormal (or normal, when possible) Choi states. We also fix a group $G$ and a $G$-action $g\mapsto\alpha_g$ on $\mf M$ and $g\mapsto\beta_g$ on $\mf N$. We denote, for every $g,\,h\in G$, by $\alpha_g\otimes\beta_h\in{\rm Aut}(\mf M\vN\mf N)$ the unique extension of the map $\mf M\times\mf N\ni(a,b)\mapsto\alpha_g(a)\otimes\beta_h(b)\in\mf M\vN\mf N$ \cite[Chapter IV, Proposition 5.13]{TakesakiI}.

\begin{definition}
We say that a channel $\Phi\in{\bf NCH}(\mf M,\mf N)$ is {\it $(\alpha,\beta)$-covariant} if $\alpha_g\circ\Phi=\Phi\circ\beta_g$ for all $g\in G$. We denote the set of $(\alpha,\beta)$-covariant channels $\Phi\in{\bf NCH}(\mf M,\mf N)$ by ${\bf NCH}_\alpha^\beta$.
\end{definition}

\begin{theorem}\label{theor:covCJ}
Assume that $\mf M$ and $\mf N$ are von Neumann algebras with $\mf M$ $\sigma$-finite and injective. Fix a faithful normal state $\rho_0$ of $\mf M$ and a GNS-construction $(\mf H,\Omega)$ for $\rho_0$ where $\Omega$ is cyclic and separating for $(\mf H,\mf M)$. A channel $\Phi\in{\bf NCH}(\mf M,\mf N)$ is $(\alpha,\beta)$-covariant if and only if we have, for the binormal Choi-Jamio\l kowski state of $\Phi$,
\begin{equation}\label{eq:covCJ}
S_\Omega^\Phi\circ(\alpha'_g\otimes\beta_g)=S_{U_{\alpha_g}^*\Omega}^\Phi,\qquad {\it for\ all}\ g\in G.
\end{equation}
If additionally, for each $g\in G$, there is $\lambda_g\geq0$ such that $\rho_0\circ\alpha_g\leq\lambda_g\rho_0$, defining $\gamma_g:=(\alpha_{g^{-1}})_{\Omega,\Omega}^\#$ for all $g\in G$, a channel $\Phi\in{\bf NCH}(\mf M,\mf N)$ is $(\alpha,\beta)$-covariant if and only if
\begin{equation}\label{eq:covCJ2}
S_\Omega^\Phi\circ(\gamma_g\otimes\beta_g)=S_\Omega^\Phi,\qquad {\it for\ all}\ g\in G.
\end{equation}
\end{theorem}

\begin{proof}
Let $\Phi\in{\bf NCH}_\alpha^\beta$. For any $g\in G$, $a'\in\mf M'$, and $b\in\mf N$,
\begin{eqnarray*}
S_\Omega^\Phi\big(\alpha'_g(a')\otimes\beta_g(b)\big)&=&\<\Omega|\alpha'_g(a')(\Phi\circ\beta_g)(b)\Omega\>=\<\Omega|\alpha'_g(a')(\alpha_g\circ\Phi)(b)\Omega\>\\
&=&\<\Omega|U_{\alpha_g}a'U_{\alpha_g}^*U_{\alpha_g}\Phi(b)U_{\alpha_g}^*\Omega\>=\<U_{\alpha_g}^*\Omega|a'\Phi(b)U_{\alpha_g}^*\Omega\>\\
&=&S_{U_{\alpha_g}^*\Omega}(a'\otimes b).
\end{eqnarray*}
Next, suppose that Equation \eqref{eq:covCJ} holds. We now have, for all $a'\in\mf M'$, $b\in\mf N$, and $g\in G$,
\begin{eqnarray*}
\<\Omega|a'(\Phi\circ\beta_g)(b)\Omega\>&=&S_\Omega^\Phi\big(a'\otimes\beta_g(b)\big)=S_{U_{\alpha_g}^*\Omega}\big(\alpha'_{g^{-1}}(a')\otimes b\big)\\
&=&\<U_{\alpha_g}^*\Omega|\alpha'_{g^{-1}}(a')\Phi(b)U_{\alpha_g}^*\Omega\>=\<\Omega|U_{\alpha_g}U_{\alpha_g}^*a'U_{\alpha_g}\Phi(b)U_{\alpha_g}^*\Omega\>\\
&=&\<\Omega|a'(\alpha_g\circ\Phi)(b)\Omega\>.
\end{eqnarray*}
Since $\Omega$ is cyclic for $(\mf H,\mf M')$, the above means that $\Phi\in{\bf NCH}_\alpha^\beta$.

Assume now that, for each $g\in G$, there is $\lambda_g\geq0$ such that $\rho_0\circ\alpha_g\leq\lambda_g\rho_0$. Using the properties of the commutant dual, we have, for all $g,\,h\in G$,
$$
\gamma_{gh}=(\alpha_{h^{-1}g^{-1}})_{\Omega,\Omega}^\#=(\alpha_{g^{-1}})_{\Omega,\Omega}^\#\circ(\alpha_{h^{-1}})_{\Omega,\Omega}^\#=\gamma_g\otimes\gamma_h.
$$
Moreover, $\gamma_g$ is invertible by $\gamma_{g^{-1}}$ since $({\rm id}_{\mf M})_{\Omega,\Omega}^\#={\rm id}_{\mf M'}$, as one easily checks. Let $\Phi\in{\bf NCH}_\alpha^\beta$. We have, for all $g\in G$, $a'\in\mf M'$, and $b\in\mf N$,
\begin{eqnarray*}
S\big(\gamma_g(a')\otimes\beta_g(b)\big)=\<\Omega|\gamma_g(a')(\Phi\circ\beta_g)(b)\Omega\>=\<\Omega|a'(\alpha_{g^{-1}}\circ\Phi\circ\beta_g)(b)\Omega\>&=&\<\Omega|a'\Phi(b)\Omega\>\\
&=&S(a'\otimes b),
\end{eqnarray*}
implying Equation \eqref{eq:covCJ2}. If, on the other hand, Equation \eqref{eq:covCJ2} holds for $\Phi\in{\bf NCH}(\mf M,\mf N)$, we have, for all $g\in G$, $a'\in\mf M'$, and $b\in\mf N$,
\begin{eqnarray*}
\<\Omega|a'(\Phi\circ\beta_g)(b)\Omega\>=S\big(a'\otimes\beta_g(b)\big)=S\big(\gamma_{g^{-1}}(a')\otimes b\big)&=&\<\Omega|\gamma_{g^{-1}}(a')\Phi(b)\Omega\>\\
&=&\<\Omega|a'(\alpha_g\circ\Phi)(b)\Omega\>,
\end{eqnarray*}
implying that $\Phi\in{\bf NCH}_\alpha^\beta$.
\end{proof}

Note that $g\mapsto\gamma_g$ of the above theorem is not an action since $\gamma_g$ typically fails to be an automorphism, although, as shown in the proof above, $g\mapsto\gamma_g$ could be called as a one-parameter group of normal completely positive maps as, for each $g\in G$, $\gamma_g$ is invertible by $\gamma_{g^{-1}}$, as one easily checks. It can already be gleaned from both equations \eqref{eq:covCJ} and \eqref{eq:covCJ2} that in the invariant case $\rho_0\alpha_g=\rho_0$, for all $g\in G$, the covariance condition significantly simplifies. This is highlighted in the following section (Section \ref{sec:invariant}).

Let us briefly discuss the case where $\mf N=\mc B$ is a unital $C^*$-algebra and we concentrate on the non-(bi)normal Choi-Jamio\l kowski isomorphism ${\bf CH}(\mf M,\mc B)\ni\Phi\mapsto S_\Omega^\Phi\in{\bf S}_{\rho_0}(\mf M'\tmin\mc B)$. We further assume that $g\mapsto\alpha_g$ is an action of a group $G$ on $\mf M$ and $G\ni g\mapsto\beta_g\in{\rm Aut}(\mc B)$ is a homomorphism where ${\rm Aut}(\mc B)$ is the group of *-automorphisms of $\mc B$. The obvious modification of Theorem \ref{theor:covCJ} characterizes $(\alpha,\beta)$-covariant channels, i.e.,\ channels $\Phi\in{\bf CH}(\mf M,\mc B)$ such that $\Phi\circ\beta_g=\alpha_g\circ\Phi$ for all $g\in G$. In this context, $\alpha'_g\otimes\beta_h\in{\rm Aut}(\mf M\tmin\mc B)$ is the unique extension of the map $\mf M'\times\mc B\ni(a',b)\mapsto\alpha'_g(a')\otimes\beta_h(b)\in\mf M\tmin\mc B$ for all $g,\,h\in G$ \cite[Chapter IV, Proposition 4.23]{TakesakiI}. In the case of the normal Choi-Jamio\l kowski isomorphism (when attainable), the result is still essentially the same.

\subsection{Covariant quantum channels whose Heisenberg output algebra is a type-I factor}

Our task now is to formulate Theorem \ref{theor:covCJ} in the case of a type-I factor as the Heisenberg output algebra. In the general case, the result (Proposition \ref{prop:vaikea}) is, however, somewhat complicated. In the following section, we will see that in a great number of cases, the covariance condition for normal channels can be significantly simplified.

We retain the assumptions made above and assume, further, that $\mf M=\mc L(\hil)$ with some separable Hilbert space $\hil$; as before, we now have access to the fully normal Choi-Jamio\l kowski isomorphism. Let $\Omega\in\hil\otimes\hil$ be a cyclic and separating vector for $\big(\hil\otimes\hil,\mc L(\hil)\otimes\C\id_\hil\big)$ and define $\rho_0\in\mc S(\hil)$, $\tr{\rho_0A}=\<\Omega|(A\otimes\id_\hil)\Omega\>$ for all $A\in\mc L(\hil)$. We may freely assume that there is an orthonormal basis $K\subset\hil$ and $t_\xi>0$ for all $\xi\in K$ such that $\sum_{\xi\in K}t_\xi=1$ and $\Omega=\sum_{\xi\in K}\sqrt{t_\xi}\xi\otimes\xi$ implying that $\rho_0=\sum_{\xi\in K}t_\xi|\xi\>\<\xi|$. Recall that $j_\hil(A)=A^{T\,*}=\ovl A$ for all $A\in\mc L(\hil)$, where the transpose is defined with respect to $K$. All automorphisms on $\mc L(\hil)$ are inner meaning that, whenever $\alpha\in{\rm Aut}\big(\mc L(\hil)\big)$ there is a unitary $U\in\mc U(\hil)$ unique up to a phase factor such that $\alpha(A)=UAU^*$ for all $A\in\mc L(\hil)$. The commutant action $\alpha'\in{\rm Aut}\big(\mc L(\hil)'\big)={\rm Aut}\big(\mc L(\hil)\big)$ is now given by
$$
\alpha'(A')=\ovl{U\ovl{A'}U^*}=\ovl UA'\ovl U^*,\qquad A'\in\mc L(\hil).
$$
It follows that, for the representation ${\rm Aut}\big(\mc L(\hil)\big)\ni\alpha\mapsto U_\alpha\in\mc U(\hil\otimes\hil)$ fixed by $\Omega$ according to \cite{Araki74}, we have $U_\alpha=U\otimes\ovl U$, when $\alpha(A)=UAU^*$ for all $A\in\mc L(\hil)$ with some $U\in\mc U(\hil)$.

Let $G$ be a group and $G\ni g\mapsto\alpha_g\in{\rm Aut}\big(\mc L(\hil)\big)$ an action. Hence, it follows that there is a projective unitary representation $U:G\to\mc U(\hil)$ such that $\alpha_g(A)=U(g)AU(g)^*$ for all $g\in G$ and $A\in\mc L(\hil)$ \cite[Theorem 7.5]{Varadarajan}. We have $\alpha'_g(A')=\ovl{U(g)}A'\ovl{U(g)}^*$ for all $g\in G$ and $A'\in\mc L(\hil)$. Let $\mf N$ be another von Neumann algebra and $G\ni g\mapsto\beta_g\in{\rm Aut}(\mf N)$ an action. According to the above and Theorem \ref{theor:covCJ}, we find that $\Phi\in{\bf NCH}\big(\mc L(\hil),\mf N\big)$ is $(\alpha,\beta)$-covariant if and only if
$$
S_\Omega^\Phi\big(\ovl{U(g)}A'\ovl{U(g)}^*\otimes\beta_g(b)\big)=S_{\big(U(g)\otimes\ovl{U(g)}\big)^*\Omega}^\Phi(A'\otimes b),\qquad A'\in\mc L(\hil),\quad b\in\mf N,\quad g\in G.
$$

We may also formulate the following somewhat more explicit characterization of covariant channels.

\begin{proposition}\label{prop:vaikea}
Suppose that $\hil$ is a separable Hilbert space and $\mf N$ is a von Neumann algebra. Fix an orthonormal basis $K\subset\hil$ and numbers $t_\xi>0$, $\xi\in K$, such that $\sum_{\xi\in K}t_\xi=1$. Denote by $\mc F$ the family of finite subsets of $K$ and, for each $F\in\mc F$, denote $\rho_F:=\sum_{\xi\in F}t_\xi|\xi\>\<\xi|$, $\rho_F^{-1/2}:=\sum_{\xi\in F}t_\xi^{-1/2}|\xi\>\<\xi|$, and $\rho_K=:\rho_0$. We treat $\mc F\times\mc F$ as a directed set with respect to set inclusion. Suppose that $G\ni g\mapsto\alpha_g\in{\rm Aut}\big(\mc L(\hil)\big)$ and $G\ni g\mapsto\beta_g\in{\rm Aut}(\mf N)$ are actions where $\alpha_g(A)=U(g)AU(g)^*$ for all $g\in G$ and $A\in\mc L(\hil)$ with some projective unitary representation $U:G\to\mc U(\hil)$. Channels $\Phi\in{\bf NCH}_\alpha^\beta$ are in one-to-one correspondence with states $S\in{\bf NS}_{\rho_0}\big(\mc L(\hil)\vN\mf N\big)$ such that
\begin{equation}\label{eq:vaikea1}
\lim_{(F,D)\in\mc F\times\mc F}S\big(\rho_F^{-1/2}\ovl{U(g)}\rho_0^{1/2}A'\rho_0^{1/2}\ovl{U(g)}^*\rho_D^{-1/2}\otimes \beta_g(b)\big)=S(A'\otimes b)
\end{equation}
for all $g\in G$, $A'\in\mc L(\hil)$, and $b\in\mf N$. The correspondence is set by the equation
\begin{equation}\label{eq:vaikea2}
S(A'\otimes b)=\tr{\rho_0^{1/2}(A')^T\rho_0^{1/2}\Phi(b)},\qquad A'\in\mc L(\hil),\quad b\in\mf N,
\end{equation}
where the transpose is taken with respect to $K$.
\end{proposition}

\begin{proof}
Equation \eqref{eq:vaikea2} is obtained in the same way as Equation \eqref{eq:recovery2} from the normal Choi-Jamio\l kowsky isomorphism associated with the vector $\Omega=\sum_{\xi\in K}\sqrt{t_\xi}\xi\otimes\xi$. It is simple to check that $(\rho_0^{1/2}\rho_F^{-1/2})_{F\in\mc F}$ converges strongly to $\id_\hil$. Suppose that $\Phi\in{\bf NCH}_\alpha^\beta$ and denote $S=S_\Omega^\Phi$. We have, for all $g\in G$, $A'\in\mc L(\hil)$, and $b\in\mf N$,
\begin{eqnarray*}
S(A'\otimes b)&=&\tr{\rho_0^{1/2}(A')^T\rho_0^{1/2}\Phi(b)}=\tr{U(g)\rho_0^{1/2}(A')^T\rho_0^{1/2}U(g)^*(\Phi\circ\beta_g)(b)}\\
&=&\lim_{(F,D)\in\mc F\times\mc F}\tr{\rho_0^{1/2}\rho_D^{-1/2}U(g)\rho_0^{1/2}(A')^T\rho_0^{1/2}U(g)^*\rho_F^{-1/2}\rho_0^{1/2}(\Phi\circ\beta_g)(b)}\\
&=&\lim_{(F,D)\in\mc F\times\mc F}\tr{\rho_0^{1/2}\big(\rho_F^{-1/2}\ovl{U(g)}\rho_0^{1/2}A'\rho_0^{1/2}\ovl{U(g)}^*\rho_D^{-1/2}\big)^T\rho_0^{1/2}(\Phi\circ\beta_g)(b)}\\
&=&\lim_{(F,D)\in\mc F\times\mc F}S\big(\rho_F^{-1/2}\ovl{U(g)}\rho_0^{1/2}A'\rho_0^{1/2}\ovl{U(g)}^*\rho_D^{-1/2}\otimes\beta_g(b)\big).
\end{eqnarray*}
Thus, Equation \eqref{eq:vaikea1} holds. Note that we use the strong convergence $\rho_0^{1/2}\rho_F^{-1/2}\overset{F\in\mc F}{\rightarrow}\id_\hil$ in the third equality above.

Suppose that Equation \eqref{eq:vaikea1} holds for all $g\in G$, $A'\in\mc L(\hil)$, and $b\in\mf N$, for a state $S\in{\bf NS}_{\rho_0}\big(\mc L(\hil)\vN\mf N\big)$, let the channel $\Phi\in{\bf NCH}\big(\mc L(\hil),\mf N\big)$ be such that $S=S_\Omega^\Phi$. We have, for all $A'\in\mc L(\hil)$, $b\in\mf N$, and $g\in G$,
\begin{eqnarray*}
\tr{\rho_0^{1/2}A'\rho_0^{1/2}(\Phi\circ\beta_g)(b)}&=&S\big((A')^T\otimes\beta_g(b)\big)\\
&=&\lim_{(F,D)\in\mc F\times\mc F}S\big(\rho_F^{-1/2}\ovl{U(g)}^*\rho_0^{1/2}(A')^T\rho_0^{1/2}\ovl{U(g)}\rho_D^{-1/2}\otimes b\big)\\
&=&\lim_{(F,D)\in\mc F\times\mc F}\tr{\rho_0^{1/2}\big(\rho_F^{-1/2}\ovl{U(g)}^*\rho_0^{1/2}(A')^T\rho_0^{1/2}\ovl{U(g)}\rho_D^{-1/2}\big)^T\rho_0^{1/2}\Phi(b)}\\
&=&\lim_{(F,D)\in\mc F\times\mc F}\tr{\rho_0^{1/2}\rho_D^{-1/2}U(g)^*\rho_0^{1/2}A'\rho_0^{1/2}U(g)\rho_F^{-1/2}\rho_0^{1/2}\Phi(b)}\\
&=&\tr{U(g)^*\rho_0^{1/2}A'\rho_0^{1/2}U(g)\Phi(b)}=\tr{\rho_0^{1/2}A'\rho_0^{1/2}(\alpha_g\circ\Phi)(b)},
\end{eqnarray*}
implying that $\Phi\in{\bf NCH}_\alpha^\beta$. We have used the strong convergence $\rho_0^{1/2}\rho_F^{-1/2}\overset{F\in\mc F}{\rightarrow}\id_\hil$ in the second-to-last equality.
\end{proof}

The limit \eqref{eq:vaikea1} in the above result makes this covariance conditions somewhat difficult check. However, if we have access to a faithful state which is invariant under the action $g\mapsto\alpha_g$, we may formulate a very simple necessary and sufficient condition for covariance. This will be the topic of the following section.

Let us retain the assumptions of Proposition \ref{prop:vaikea}. If there is a Hilbert space $\mc K$ such that $\mf N=\mc L(\mc K)$ whence there is a projective unitary representation $V:G\to\mc U(\mc K)$ such that $\beta_g(B)=V(g)BV(g)^*$ for all $g\in G$ and $B\in\mc L(\mc K)$, $\Phi\in{\bf NCH}(\hil,\mc K)$ is $(\alpha,\beta)$-covariant if and only if
\begin{equation}\label{eq:CovChanChar}
\big(\ovl{U(g)}\otimes V(g)\big)^*S_\Omega^\Phi\big(\ovl{U(g)}\otimes V(g)\big)=S_{\big(U(g)\otimes\ovl{U(g)}\big)\Omega}
^\Phi,\qquad g\in G.
\end{equation}
For $S:=S_\Omega^\Phi$, the condition \eqref{eq:vaikea1} becomes
$$
\lim_{(F,D)\in\mc F\times\mc F}\big(\rho_F^{-1/2}\ovl{U(g)}\rho_0^{1/2}\otimes V(g)\big)^*S\big(\rho_D^{-1/2}\ovl{U(g)}\rho_0^{1/2}\otimes V(g)\big)=S
$$
for all $g\in G$ where the limit is taken with respect to the $\sigma$-weak topology in $\mc T(\hil\otimes\mc K)$.

If $\dim{\hil}=:d<\infty$, we may choose $\rho_0=d^{-1}\id_\hil$ which is obviously invariant under all automorphisms. Let us fix an orthonormal basis $\{|n\>\}_{n=1}^d$ and define $\Omega=d^{-1/2}\sum_{n=1}^d|n,n\>$; we denote $|m,n\>:=|m\>\otimes|n\>$ for all $m,\,n=1,\ldots,\,d$. It follows that, for all $\Phi\in{\bf NCH}(\hil,\mc K)$,
$$
S^\Phi:=S_\Omega^\Phi=\frac{1}{d}\sum_{m,n=1}^d|m\>\<n|\otimes\Phi_*(|m\>\<n|),
$$
and $\Phi\in{\bf NCH}_\alpha^\beta$ if and only if
$$
\big(\ovl{U(g)}\otimes V(g)\big)S^\Phi=S^\Phi\big(\ovl{U(g)}\otimes V(g)\big),\qquad g\in G.
$$

\section{Examples involving invariant faithful states}\label{sec:invariant}

The results of the previous section can be greatly streamlined if there is a faithful normal state on the Heisenberg output algebra which is invariant with respect to the action $G\ni g\mapsto\alpha_g$. In this section we concentrate on this special case and also study a couple of examples (the modular automorphism group and phase shifts) of this in two subsections.

We again concentrate on the case where $\mf M$ and $\mf N$ are von Neumann algebras and $\mf M$ is, additionally, $\sigma$-finite and injective. We keep the earlier notations associated with a GNS-construction $(\mf H,\Omega)$ for a faithful state $\rho_0\in{\bf NS}_1(\mf M)$ with a cyclic and separating vector $\Omega$ for $(\mf H,\mf M)$ fixed. We let $G$ be a group and $G\ni g\mapsto\alpha_g\in{\rm Aut}(\mf M)$ and $G\ni g\mapsto\beta_g\in{\rm Aut}(\mf N)$ be actions. Characterizing covariant channels becomes particularly simple when the trajectory $G\ni g\mapsto\rho_0\circ\alpha_g\in{\bf NS}_1(\mf M)$ shrinks into a singleton, i.e.,\ $\rho_0$ is $\alpha$-invariant; $\rho_0\circ\alpha_g=\rho_0$ for all $g\in G$. Since all the automorphisms $\alpha_g$, $g\in G$, belong now to the stabilizing subgroup $A_{\rho_0}$ of $\rho_0$ within ${\rm Aut}(\mf M)$, we have $U_{\alpha_g}\Omega=\Omega$ for all $g\in G$. This means that $\Phi\in{\bf NCH}(\mf M,\mf N)$ is $(\alpha,\beta)$-covariant if and only if $S_\Omega^\Phi\circ(\alpha_g\otimes\beta_g)=S_{U_{\alpha_g}^*\Omega}^\Phi=S_\Omega^\Phi$ for all $g\in G$. On the other hand, in this invariant case, the one-parameter group $g\mapsto\gamma_g$ of Theorem \ref{theor:covCJ} coincides with $g\mapsto\alpha'_g$, as one easily checks.

One particular instance where the above happens is the case where $G$ is compact and $g\mapsto\alpha_g$ is continuous. This is due to the fact, that, in this case, we may make normal states of $\mf M$ $\alpha$-invariant. Indeed, we may define the map ${\bf S}_1(\mf M)\ni\rho\mapsto\ovl\rho\in{\bf S}_1$, $\ovl\rho(a)=\int_G(\rho\circ\alpha_h)(a)\,d\mu(h)$, $a\in\mf M$, where $\mu$ is the right Haar measure of $G$ with $\mu(G)=1$. This is easily seen to be a positive map. If $\rho$ is faithful, also $\rho\circ\alpha_h$ is faithful for every $h\in G$ and, consequently, also $\ovl\rho$ is faithful. If $\rho$ is normal, using the monotone convergence theorem, one easily shows that $\ovl\rho$ is normal as well. Thus, there exists a faithful state $\rho_0\in{\bf NS}_1(\mf M)$ which is $\alpha$-invariant. When we endow this $\rho_0$ with a GNS-representation $(\mf M,\Omega)$ where $\Omega$ is cyclic and separating for $(\mf H,\mf M)$, channels $\Phi\in{\bf NCH}_\alpha^\beta$ are in one-to-one correspondence with binormal Choi-Jamio\l kowski states $S_\Omega^\Phi\in{\bf S}^{\rm bin}_{\rho_0}(\mf M'\tmin\mf N)$ such that $S_\Omega^\Phi\circ(\alpha'_g\otimes\beta_g)=S_\Omega^\Phi$ for every $g\in G$. We may combine the above observations into the following result.

\begin{theorem}\label{theor:invariant}
Suppose that $\mf M$ and $\mf N$ are von Neumann algebras where $\mf M$ is, additionally, $\sigma$-finite and injective, $G$ is a group, and $G\ni g\mapsto\alpha_g\in{\rm Aut}(\mf M)$ and $G\ni g\mapsto\beta_g\in{\rm Aut}(\mf N)$ are actions. If there is a faithful state $\rho_0\in{\bf NS}_1(\mf M)$ such that $\rho_0\circ\alpha_g=\rho_0$ for all $g\in G$, the channels $\Phi\in{\bf NCH}_\alpha^\beta$ are in one-to-one correspondence with binormal states $S\in{\bf S}^{\rm bin}_{\rho_0}(\mf M'\tmin\mf N)$ such that $S\circ(\alpha'_g\otimes\beta_g)=S$. This correspondence is mediated by the binormal Choi-Jamio\l kowsky isomorphism ${\bf NCH}(\mf M,\mf N)\ni\Phi\mapsto S_\Omega^\Phi\in{\bf S}^{\rm bin}_{\rho_0}(\mf M'\tmin\mf N)$ restricted onto ${\bf NCH}_\alpha^\beta$ and defined by any GNS-construction $(\mf H,\Omega)$ for $\rho_0$ where $\Omega$ is cyclic and separating for $(\mf H,\mf M)$. Especially this happens when $G$ is compact and $g\mapsto\alpha_g$ is continuous.
\end{theorem}

The analogue of this theorem naturally applies in the case when $\mf M$ is such that the state $\tilde{\rho_0}\in{\bf S}(\mf M'\tmin\mf M)$ introduced after the proof of Theorem \ref{theor:CJisomorphism} extends into $\hat{\rho_0}\in{\bf NS}(\mf M'\vN\mf M)$. In this case, when $\rho_0\circ\alpha_g=\rho_0$ for all $g\in G$, $(\alpha,\beta)$-covariant channels are in one-to-one correspondence with normal states $S\in{\bf NS}_{\rho_0}(\mf M'\vN\mf N)$ such that $S\circ(\alpha_g\otimes\beta_g)=S$ for all $g\in G$.

\subsection{An example: the modular automorphism group}\label{subsec:modulargroup}

Suppose that $\mf M$ is an injective $\sigma$-finite von Neumann algebra and $\mf N$ is another von Neumann algebra. Suppose that $\rho_0$ is a faithful normal state of $\mf M$ with a cyclic and separating vector $\Omega\in\mf H$ and retain the notations of the previous sections. We now assume that the symmetry group is the additive real line $\R$ and $\alpha_t(a)=\Delta^{it}a\Delta^{-it}$, $t\in\R$, $a\in\mf M$, i.e.,\ $\R\ni t\mapsto\alpha_t\in{\rm Aut}(\mf M)$ is the {\it modular automorphism group}. This action is associated with the equilibrium dynamics of a quantum system. Since $\rho_0$ is $\alpha$-invariant, it follows that, for any action $\R\ni t\mapsto\beta_t\in{\rm Aut}(\mf N)$ and any $\Phi\in{\bf NCH}_\alpha^\beta$, $S_\Omega^\Phi\circ(\alpha'_t\otimes\beta_t)=S_\Omega^\Phi$ for all $t\in\R$. Note that $\alpha'_t(a')=\Delta^{it}a'\Delta^{-it}$ for all $t\in\R$ and $a'\in\mf M'$.

Let us look at the `fully quantum' case where $\mf M=\mc L(\hil)$ and $\mf N=\mc L(\mc K)$, where $\hil$ and $\mc K$ are both separable Hilbert spaces. Now, $\rho_0\in\mc S(\hil)$ is faithful and $\alpha_t(A)=\rho_0^{it}A\rho_0^{-it}$, $t\in\R$, $A\in\mc L(\hil)$. We assume that $t\mapsto\beta_t$ is continuous, so that, as all automorphisms on a type-I factor are inner, there is a strongly continuous unitary representation $V:\R\to\mc U(\mc K)$ such that $\beta_t(B)=V(t)BV(t)^*$, $t\in\R$, $B\in\mc L(\mc K)$. According to the SNAG-theorem (Segal-Na\u{\i}mark-Ambrose-Godement theorem), there is a spectral measure $\ms P:{\bf Leb}(\R)\to\mc L(\mc K)$ on the Lebesgue-$\sigma$-algebra of $\R$ such that $V(t)=\int_\R e^{iht}\,d\ms P(h)$, $t\in\R$. Equivalently, there is a self-adjoint $H:\mc D(H)\to\mc K$ such that $V(t)=e^{itH}$, $t\in\R$. We have $\Phi\in{\bf NCH}_\alpha^\beta$ if and only if
\begin{equation}\label{eq:modularinvariant}
(\rho_0^{-it}\otimes e^{itH})S_\Omega^\Phi=S_\Omega^\Phi(\rho_0^{-it}\otimes e^{itH}),\qquad {\rm for\ all}\ t\in\R.
\end{equation}

Suppose now that also $t\mapsto\beta_t$ is a modular group. Thus, there is a faithful state $\sigma_0\in\mc S(\mc K)$ such that $\beta_t(B)=\sigma_0^{it}B\sigma_0^{-it}$, $t\in\R$, $B\in\mc L(\mc K)$. It follows that $\Phi\in{\bf NCH}_\alpha^\beta$ if and only if
$$
(\rho_0^{-1}\otimes\sigma_0)S_\Omega^\Phi=S_\Omega^\Phi(\rho_0^{-1}\otimes\sigma_0).
$$

\subsection{An example: phase-shift-covariant channels}\label{subsec:phaseshift}

We let $\hil$ be a Hilbert space with an orthonormal basis $\{|n\>\}_{n=0}^\infty$. We treat $[0,2\pi)$ as an additive cyclic group. Define the unitary representation $U:[0,2\pi)\to\mc U(\hil)$
$$
U(\tj)=\sum_{n=0}^\infty e^{in\tj}|n\>\<n|,\qquad\tj\in[0,2\pi).
$$
The Hilbert space $\hil$ describes a harmonic oscillator where $U$ mediates the phase shifts. We are interested in channels $\Phi\in{\bf NCH}(\hil,\hil)$ which are symmetric under the phase shifts, i.e.,\ $\Phi\big(U(\tj)BU(\tj)^*\big)=U(\tj)\Phi(B)U(\tj)^*$ for all $\tj\in[0,2\pi)$ and $B\in\mc L(\hil)$. We denote the set of these channels by ${\bf NCH}_U$ and call channels $\Phi\in{\bf NCH}_U$ as {\it $U$-covariant}. Because the group $[0,2\pi)$ is compact, there are invariant faithful states which are diagonalized in the basis $\{|n\>\}_{n=0}^\infty$. Let us define, for all $A\in\mc L(\hil)$, $\ovl A\in\mc L(\hil)$ through $\<m|\ovl A|n\>=\ovl{\<m|A|n\>}$ for all $m,\,n=0,\,1,\,2,\ldots$, i.e., $\ovl A=A^{T\,*}$ with the transpose defined with respect to the basis $\{|n\>\}_{n=0}^\infty$. The Choi-Jamio\l kowski states of channels $\Phi\in{\bf NCH}_U$ are those $S\in\mc L(\hil\otimes\hil)$ with a fixed first margin such that
\begin{equation}\label{eq:phasecovCJ}
\big(\ovl{U(\tj)}\otimes U(\tj)\big)S=S\big(\ovl{U(\tj)}\otimes U(\tj)\big),\qquad\tj\in[0,2\pi).
\end{equation}
We go on to characterizing ${\bf NCH}_U$ using Theorem \ref{theor:invariant}.

\begin{proposition}
A channel $\Phi\in{\bf NCH}(\hil,\hil)$ is $U$-covariant if and only if there are $\tau_{l,j,m}\in\C$, $l,\,j,\,m=0,\,1,\,2,\ldots$, such that $\sum_{l,j=0}^\infty|\tau_{l,j,m}|^2=1$ for all $m=0,\,1,\,2,\ldots$ and, upon defining $K_{l,j}\in\mc L(\hil)$, $K_{l,j}|m\>=\tau_{l,j,m}|m+l\>$, for all $l,\,j,\,m=0,\,1,\,2,\ldots$,
\begin{equation}\label{eq:Ucovchar}
\Phi(B)=\sum_{l,j=0}^\infty K_{l,j}^*BK_{l,j},\qquad B\in\mc L(\hil).
\end{equation}
Moreover, when $\tau_{l,j,m}\in\C$, $l,\,j,\,m=0,\,1,\,2,\ldots$ are such that $\sum_{l,j=0}^\infty|\tau_{l,j,m}|^2=1$ for all $m=0,\,1,\,2,\ldots$, defining $K_{l,j}\in\mc L(\hil)$, $K_{l,j}|m\>=\tau_{l,j,m}|m+l\>$, for all $l,\,j,\,m=0,\,1,\,2,\ldots$, Equation \eqref{eq:Ucovchar} sets up a channel $\Phi\in{\bf NCH}_U$.
\end{proposition}

\begin{proof}
Let $\Phi\in{\bf NCH}_U$ and $S:=S_\Omega^\Phi$ with $\Omega=\sum_{n=0}^\infty \sqrt{t_n}|n,n\>$ where $t_n>0$ for all $n=0,\,1,\,2,\ldots$ and $\sum_{n=0}^\infty t_n=1$; this vector provides a GNS-representation for $\rho_0\in\mc S(\hil)$, $\rho_0=\sum_{n=0}^\infty t_n|n\>\<n|$, and it is cyclic and separating for $\big(\hil\otimes\hil,\mc L(\hil)\otimes\C\id_\hil\big)$. One sees immediately that $\big(U(\tj)\otimes\ovl{U(\tj)}\big)^*\Omega=\Omega$ for all $\tj\in[0,2\pi)$ as it should be. The $U$-covariance of $\Phi$ is equivalent with Equation \eqref{eq:phasecovCJ}. In order to characterize $S$, we have to decompose $\tj\mapsto\ovl{U(\tj)}\otimes U(\tj)$ into irreducibles. It is quite clear that the correct decomposition is
$$
\ovl{U(\tj)}\otimes U(\tj)=\sum_{l=0}^\infty e^{il\tj}\bigg(\sum_{m=0}^\infty|m,m+l\>\<m,m+l|\bigg),\qquad\tj\in[0,2\pi).
$$
Denote, for each $l=0,\,1,\,2,\ldots$, the Hilbert space spanned by $|m,m+l\>$, $m=0,\,1,\,2,\ldots$, by $\mc K_l$. Hence, $\hil\otimes\hil=\bigoplus_{l=0}^\infty\mc K_l$ and $S=\bigoplus_{l=0}^\infty S_l$ with some positive $S_l\in\mc T(\mc K_l)$, $l=0,\,1,\,2,\ldots$ It follows after a simple calculation that
$$
\rho_0={\rm tr}_2[S]=\sum_{m=0}^\infty\sum_{l=0}^\infty\<m,m+l|S_l|m,m+l\>|m\>\<m|,
$$
implying that $t_m=\sum_{l=0}^\infty\<m,m+l|S_l|m,m+l\>$ for all $m=0,\,1,\,2,\ldots$. Since, for every $l=0,\,1,\,2,\ldots$, $S_l$ are positive, it follows that there is a Hilbert-Schmidt operator $T_l$ on $\mc K_l$ such that $S_l=T_l^*T_l$ or, upon denoting $\tilde{\tau}_{l,j,m}:=\ovl{\<j,j+l|T_l|n,n+l\>}$ for all $j,\,m=0,\,1,\,2,\ldots$, $\<n,n+l|S_l|m,m+l\>=\sum_{j=0}^\infty\ovl{\tilde{\tau}_{l,j,m}}\tilde{\tau}_{l,j,n}$, $m,\,n=0,\,1,\,2,\ldots$. Define $\tau_{l,j,m}=t_m^{-1/2}\tilde{\tau}_{l,j,m}$, $l,\,j,\,m=0,\,1,\,2,\ldots$. We have, for all $m=0,\,1,\,2,\ldots$,
$$
\sum_{l,j=0}^\infty|\tau_{l,j,m}|^2=\frac{1}{t_m}\sum_{l,j=0}^\infty\<m,m+l|S_l|m,m+l\>=\frac{t_m}{t_m}=1
$$
and, defining $K_{l,j}\in\mc L(\hil)$, $l,\,j=0,\,1,\,2,\ldots$, as in the claim, we may calculate, using Equation \eqref{eq:recovery2}, for all $m,\,n=0,\,1,\,2,\ldots$,
\begin{eqnarray*}
\<m|\Phi(B)|n\>&=&\frac{1}{\sqrt{t_m t_n}}\tr{S(|m\>\<n|\otimes B)}=\frac{1}{\sqrt{t_m t_n}}\sum_{l,r=0}^\infty\<r,r+l|S_l(|m\>\<n|\otimes B)|r,r+l\>\\
&=&\frac{1}{\sqrt{t_m t_n}}\sum_{l=0}^\infty\<n,n+l|S_l(\id_\hil\otimes B)|m,n+l\>\\
&=&\frac{1}{\sqrt{t_m t_n}}\sum_{l,r=0}^\infty\<n,n+l|S_l|r,r+l\>\<r,r+l|\id_\hil\otimes B|m,n+l\>\\
&=&\frac{1}{\sqrt{t_m t_n}}\sum_{l=0}^\infty\<n,n+l|S_l|m,m+l\>\<m+l|B|n+l\>\\
&=&\frac{1}{\sqrt{t_m t_n}}\sum_{l,j=0}^\infty\ovl{\tilde{\tau}_{l,j,m}}\tilde{\tau}_{l,j,n}\<m+l|B|n+l\>\\
&=&\sum_{l,j=0}^\infty\ovl{\tau_{l,j,m}}\tau_{l,j,n}\<m+l|B|n+l\>=\sum_{l,j=0}^\infty\<m|K_{l,j}^*BK_{l,j}|n\>.
\end{eqnarray*}

Suppose then that $\tau_{l,j,m}\in\C$, $l,\,j,\,m=0,\,1,\,2,\ldots$ are such that, for all $m=0,\,1,\,2,\ldots$, $\sum_{l,j=0}^\infty|\tau_{l,j,m}|^2=1$, and define the operators $K_{l,j}$, $l,\,j=0,\,1,\,2,\ldots$, as in the claim. Let us show that Equation \eqref{eq:Ucovchar} truly defines $\Phi\in{\bf NCH}_U$. We have, for any $m,\,n=0,\,1,\,2,\ldots$,
$$
\<m|\Phi(\id_\hil)|n\>=\sum_{l,j=0}^\infty\<m|K_{l,j}^*K_{l,j}|n\>=\sum_{l,j=0}^\infty\ovl{\tau_{l,j,m}}\tau_{l,j,n}\<m+l|n+l\>=\<m|n\>\sum_{l,j=0}^\infty|\tau_{l,j,m}|^2=\<m|n\>,
$$
implying that $\Phi$ is unital. Clearly, $\Phi$ is normal and completely positive (since we have its Kraus decomposition). Let $m,\,n=0,\,1,\,2,\ldots$ and $\tj\in[0,2\pi)$. We have, for all $B\in\mc L(\hil)$,
\begin{eqnarray*}
\<m|\Phi\big(U(\tj)BU(\tj)^*\big)|n\>&=&\sum_{l,j=0}^\infty\ovl{\tau_{l,j,m}}\tau_{l,j,n}\<m+l|U(\tj)BU(\tj)^*|n+l\>\\
&=&e^{i(m-n)\tj}\sum_{l,j=0}^\infty\ovl{\tau_{l,j,m}}\tau_{l,j,n}\<m+l|B|n+l\>\\
&=&e^{i(m-n)\tj}\<m|\Phi(B)|n\>=\<m|U(\tj)\Phi(B)U(\tj)^*|n\>,
\end{eqnarray*}
implying that $\Phi$ is $U$-covariant. This finalizes the proof.
\end{proof}

In particular, setting $\tau_{0,0,m}=e^{im\tj_0}$ for some fixed $\tj_0\in[0,2\pi)$ and all $m=0,\,1,\,2,\ldots$ and $\tau_{l,j,m}=0$ for all $m=0,\,1,\,2,\ldots$ whenever $l,\,j\neq0$, the Equation \eqref{eq:Ucovchar} defines the channel $\Phi_{\tj_0}\in{\bf NCH}_U$, $\Phi_{\tj_0}(B)=U(\tj_0)^*BU(\tj_0)$, $B\in\mc L(\hil)$.

\section{Faithful states with proper orbits}

We finally discuss the case where there is no invariant faithful state available for us and how we can simplify the general problem using the results of the preceding section (Theorem \ref{theor:invariant}). We concentrate on the case Euclidean symmetries of rigid motions and lay some groundwork for an exhaustive determination of Euclidean-covariant channels. This full characterization is not yet reached in this work, but we illustrate what steps are to be taken for this goal.

We again let $\mf M$ and $\mf N$ be von Neumann algebras where $\mf M$ is, additionally, $\sigma$-finite and injective. Pick a faithful state $\rho_0\in{\bf NS}_1(\mf M)$ and let $(\mf H,\Omega)$ be a GNS-construction for $\rho_0$ where $\Omega$ is cyclic and separating for $(\mf H,\mf M)$. We define the modular structures and ${\rm Aut}(\mf M)\ni\alpha\mapsto U_\alpha\in\mc U(\mf H)$ in the same way as earlier. We also assume that $G$ is a locally compact $\sigma$-compact group and that $G\ni g\mapsto\alpha_g\in{\rm Aut}(\mf M)$ and $G\ni g\mapsto\beta_g\in{\rm Aut}(\mf N)$ are actions where the first one is continuous.

Denote the subgroup of those $h\in G$ leaving $\rho_0$ invariant by $H$, i.e.,\ using our earlier observations due to the definition of the representation ${\rm Aut}(\mf M)\ni\alpha\mapsto U_\alpha\in\mc U(\mf H)$,
$$
H=\{h\in G\,|\,\rho_0\circ\alpha_h=\rho_0\}=\{h\in G\,|\,U_{\alpha_h}^*\Omega=\Omega\}.
$$
As in the case of the total stabilizer group $A_{\rho_0}\leq{\rm Aut}(\mf M)$, also $H$ is closed. Due to the $\sigma$-compactness of $G$, the orbit $\mc O_{\rho_0}:=\{\rho\circ\alpha_{g^{-1}}\,|\,g\in G\}$ is bijectively homeomorphic with the left coset space $G/H$. Equivalently, $\mc O_\Omega:=\{U_{\alpha_g}\Omega\,|\,g\in G\}$ is bijectively homeomorphic with $G/H$.

We have seen that the channels which are covariant with respect to a compact group can, in principle, be quite conveniently characterized according to Theorem \ref{theor:invariant} with the binormal Choi-Jamio\l kowsky method. The general case is trickier, as seen in Proposition \ref{prop:vaikea}. However, in most physical cases, the Choi-Jamio\l kowsky method can be used to simplify the task of characterizing covariant channels. To see this, let us retain the assumptions made in this section thus far. Let us assume, additionally that there is a closed normal subgroup $X\leq G$ such that $G=XH$ and $X\cap H=\{e\}$ where $e$ is the neutral element of $G$; a prototypical example is a semidirect product $G=H\times_\delta X$ where $H\ni h\mapsto\delta_h\in{\rm Aut}(X)$ is a homomorphism. Define the actions $H\ni h\mapsto\tilde{\alpha}_h\in{\rm Aut}(\mf M)$, $H\ni h\mapsto\tilde{\beta}_h\in{\rm Aut}(\mf N)$, $X\ni x\mapsto\hat{\alpha}_x\in{\rm Aut}(\mf M)$, and $X\ni x\mapsto\hat{\beta}_x\in{\rm Aut}(\mf N)$ where $\tilde{\alpha}_h=\alpha_h$ and $\tilde{\beta}_h=\beta_h$ for all $h\in H$ and $\hat{\alpha}_x=\alpha_x$ and $\hat{\beta}_x=\beta_x$ for all $x\in X$. It is simple to check that ${\bf NCH}_\alpha^\beta={\bf NCH}_{\tilde{\alpha}}^{\tilde{\beta}}\cap{\bf NCH}_{\hat{\alpha}}^{\hat{\beta}}$. Thus $(\alpha,\beta)$-covariant channels are $(\tilde{\alpha},\tilde{\beta})$-covariant channels, which can be characterized according to Theorem \ref{theor:invariant}, and which are, additionally, $(\hat{\alpha},\hat{\beta})$-covariant. The last condition can be studied by using methods of Proposition \ref{prop:vaikea} or, e.g.,\ dilation techniques of \cite{HaPe2017}.

\subsection{An example: Euclidean covariance}\label{subsec:euclidean}

We go on to studying channels which are covariant with respect to the Euclidean group in $\R^3$. We are not yet, however, in the position to thoroughly characterize such channels, but this example illustrates how one can properly choose a faithful state so as to simplify the characterization of these covariant channels. In this example $\mf M=\mc L(\hil)$, $\hil=L^2(\R^3)\otimes\C^{2j+1}$ where $L^2(\R^3)$ is the Hilbert space of (equivalence classes of) Lebesgue-square-integrable functions $\fii:\R^3\to\C$, $j\in\{0,1/2,2,3/2,\ldots\}$, and $\hil$ is viewed as a set of vector fields $\fii:\R^3\to\C^{2j+1}$ in the natural way. Also $\mf N=\mc L(\mc K)$ for some Hilbert space $\mc K$.

Denote by $\mf S$ the real linear span of the Pauli matrices $\sigma_n\in\mc L(\C^2)$, $n=1,\,2,\,3$, characterized by $\sigma_l\sigma_m=\sum_{n=1}^3\eps_{l,m,n}\sigma_n$ where $\eps_{l,m,n}$, $l,\,m,\,n=1,\,2,\,3$ is the Levi-Civit\`a symbol. The map $m:\R^3\to\mf S$, $m(x_1,x_2,x_3)=x_1\sigma_1+x_2\sigma_2+x_3\sigma_3$, $(x_1,x_2,x_3)\in\R^3$, is a bijection. We define the covering homomorphism $\delta:SU(2)\to SO(3)$ in the usual way, i.e.,\ $m\big(\delta(U)\vec{x}\big)=Um(\vec{x})U^*$ for all $U\in SU(2)$ and $\vec{x}\in\R^3$. We let $G$ be the covering group of the Euclidean group of rigid motions in $\R^3$, i.e.,\ $G$ is the semidirect product of $SU(2)$ and $\R^3$, the group law given by
$$
(U,\vec{a})(V,\vec{b})=(UV,\vec{a}+\delta(U)\vec{b}),\qquad U,\,V\in SU(2),\quad\vec{a},\,\vec{b}\in\R^3.
$$
Pick the $(2j+1)$-dimensional irreducible representation $D^j$ of $SU(2)$ and define the irreducible unitary representation $Z:G\to\mc U(\hil)$,
$$
\big(Z(U,\vec{a})\fii\big)(\vec{x})=D^j(U)\fii\big(\delta(U)^T(\vec{x}-\vec{a})\big),\qquad(U,\vec{a})\in G,\quad\fii\in\hil,\quad\vec{x}\in\R^3.
$$
The action $g\mapsto\alpha_g$ is now defined through $\alpha_{(U,\vec{a})}(A)=Z(U,\vec{a})AZ(U,\vec{a})^*$, $(U,\vec{a})\in G$, $A\in\mc L(\hil)$. It follows that the Hilbert space $\hil$ describes an irreducible Euclidean-invariant spin-$j$ quantum object. We also fix an action $G\ni g\mapsto\beta_g\in{\rm Aut}\big(\mc L(\mc K)\big)$, meaning that there is a projective unitary representation $V:G\to\mc U(\mc K)$ such that $\beta_g(B)=V(g)BV(g)^*$, $g\in G$, $B\in\mc L(\mc K)$.

According to the beginning of Section \ref{sec:invariant}, there is a faithful state $\rho_0\in\mc S(\hil)$ such that $Z(U,0)\rho_0=\rho_0 Z(U,0)$ for all $U\in SU(2)$. Let $(\hil\otimes\hil,\Omega)$ be a GNS-construction for $\rho_0$ where $\Omega$ is cyclic and separating for $\big(\hil\otimes\hil,\mc L(\hil)\otimes\C\id_\hil\big)$ and $U_{\alpha_g}=Z(g)\otimes\ovl{Z(g)}$. The operation $A\mapsto\ovl A$ is defined with respect to the basis chosen when picking $\Omega$; note that the basis can be chosen so that also $U\mapsto Z(U,0)$ is decomposed into irreducibles in the same basis.. It follows that the orbit $\{Z(g)\rho_0 Z(g)^*\,|\,g\in G\}$ can be identified with $\R^3$. Denote $Z(\id,\vec{a})=:\lambda(\vec{a})$ for all $\vec{a}\in\R^3$. According to Equation \eqref{eq:CovChanChar}, the covariant channels $\Phi\in{\bf NCH}_\alpha^\beta$ are now characterized by the condition
$$
\big(Z(U,\vec{a})\otimes\ovl{Z(U,\vec{a})}\big)^*S_\Omega^\Phi\big(Z(U,\vec{a})\otimes\ovl{Z(U,\vec{a})}\big)=S_{\lambda(\vec{a})\otimes\ovl{\big(\lambda(\vec{a})}\big)\Omega}^\Phi,\qquad(U,\vec{a})\in G.
$$
The $(\alpha,\beta)$ covariance prescribes channels that translate the symmetries described by $V$ into symmetries described by $Z$. Especially, if $V=Z$, The channels $\Phi\in{\bf NCH}_\alpha^\beta={\bf NCH}_\alpha^\alpha$ correspond to transformations of the spin-$j$ object that conserve the Euclidean symmetry.

Let us take a closer look at the simplest spin-0 case, i.e.,\ $j=0$. We may now simplify the representation $Z$ to obtain $Z_0:G_0\to\mc U(\hil)$, where $G_0$ is the semidirect product of $SO(3)$ and $\R^3$ with the group law described by
$$
(Q,\vec{a})(R,\vec{b})=(QR,\vec{a}+Q\vec{b}),\qquad\vec{a},\,\vec{b}\in\R^3,\quad Q,\,R\in SO(3),
$$
$\hil=L^2(\R^3)$, and $\big(Z_0(R,\vec{a})\fii\big)(\vec{x})=\fii\big(R^T(\vec{x}-\vec{a})\big)$, $(R,\vec{a})\in G_0$, $\fii\in\hil$, $\vec{x}\in\R^3$. For simplicity, we assume that $V$ is defined on $G_0$ as well.

Denote by $\mb S^2$ the unit sphere in $\R^3$ and by $L^2(\mb S^2)$ the Hilbert space of (equivalence classes of) functions $\eta:\mb S^2\to\C$ which are square integrable with respect to the natural angular-variable measure $\sin\tj\,d\fii\,d\tj$. The space $L^2(\mb S^2)$ has the orthonormal basis $\{Y_{l,m}\,|\,m=-l,\ldots,\,l,\ l=0,\,1,\,2,\ldots\}$ consisting of the angular harmonic functions,
$$
Y_{l,m}(\fii,\tj)=(-1)^m\sqrt{\frac{2l+1}{4\pi}\frac{(l-|m|)!}{(1+|m|)!}}P_l^{|m|}(\cos\tj)e^{im\fii},\quad(\fii,\tj)\in\mb S^2,
$$
where $P_l^n:[-1,1]\to\R$ are the associated Legendre polynomials,
$$
P_l^n(t)=(-1)^{l+n}\frac{(l+n)!}{(l-n)!}(1-t^2)^{-n/2}2^l l!\frac{d^{l-n}}{dt^{l-n}}\big((1-t^2)^l\big),\quad -1\leq t\leq1.
$$
The subspaces $\mc K_l:={\rm span}\{Y_{l,m}\}_{m=-l}^l$, $l=0,\,1,\,2,\ldots$, are invariant under $SO(3)\ni R\mapsto Z_0(R,0)$ and the restriction $D^l$ of this representation onto the subspace $\mc K_l$ is irreducible for all $l=0,\,1,\,2,\ldots$.

Let $L^2(r^2\,dr)$ be the Hilbert space of (equivalence classes of) functions $f:[0,\infty)\to\C$ such that $\int_0^\infty f(r)r^2\,dr<\infty$ where $dr$ is the restriction of the Lebesgue measure on $[0,\infty)$. It follows that we may decompose $\hil$ and the representation $R\mapsto Z_0(R,0)$ as follows:
\begin{eqnarray*}
\hil&=&\bigoplus_{l=0}^\infty\big(\mc K_l\otimes L^2(r^2\,dr)\big),\\
Z_0(R,0)&=&\bigoplus_{l=0}^\infty\big(D^l(R)\otimes\id_{L^2(r^2\,dr)}\big),\qquad R\in SO(3).
\end{eqnarray*}
Any state $\rho_0\in\mc S(\hil)$ such that $Z_0(R,0)\rho_0=\rho_0Z_0(R,0)$ for all $R\in SO(3)$ are of the form
$$
\rho_0=\bigoplus_{l=0}^\infty\frac{t_l}{2l+1}\id_{\mc K_l}\otimes\sigma_l
$$
for some $t_l\geq0$, $l=0,\,1,\,2,\ldots$, such that $\sum_{l=0}^\infty t_l=1$ and some states $\sigma_l\in\mc S\big(L^2(r^2\,dr)\big)$.

Suppose now that $t_l>0$ and $\sigma_l$ is faithful for all $l=0,\,1,\,2,\ldots$. Suppose that $A\in\mc L(\hil)$, $A\geq0$, and $\tr{\rho_0 A}=0$. Denote, for all $l=0,\,1,\,2,\ldots$, by $P_l$ the orthogonal projection of $\hil$ onto the subspace $\mc K_l\otimes L^2(r^2\,dr)$. It follows that
$$
0=\tr{\rho_0 A}=\sum_{l=0}^\infty\tr{P_l\rho_0P_lA}=\sum_{l=0}^\infty\frac{t_l}{2l+1}\tr{(\id_{\mc K_l}\otimes\sigma_l)P_lAP_l}.
$$
Since $(2l+1)^{-1}\id_{\mc K_l}\otimes\sigma_l$ is faithful for all $l=0,\,1,\,2,\ldots$, as one easily shows, it follows that $P_lAP_l=0$ for every $l=0,\,1,\,2,\ldots$. Suppose now that $l\neq k$ and define $\fii_{\kappa,\lambda}=\kappa\fii_k+\lambda\fii_l\in\hil$ for all $\kappa,\,\lambda\in\C$ where $P_i\fii_i=\fii_i$, $i=k,\,l$. We have $\<\fii_{\kappa,\lambda}|A\fii_{\kappa,\lambda}\>$ for all $\kappa,\,\lambda\in\C$, which equals with
\begin{eqnarray*}
0&\leq&\left(\begin{array}{cc}
\<\fii_k|P_kAP_k\fii_k\>&\<\fii_k|P_kAP_l\fii_l\>\\
\<\fii_l|P_lAP_k\fii_k\>&\<\fii_l|P_lAP_l\fii_l\>
\end{array}\right)=\left(\begin{array}{cc}
0&\<\fii_k|P_kAP_l\fii_l\>\\
\<\fii_l|P_lAP_k\fii_k\>&0
\end{array}\right)\\
&\Leftrightarrow&\<\fii_k|P_kAP_l\fii_l\>=0,
\end{eqnarray*}
implying that, as we vary $\fii_k$ and $\fii_l$, $P_kAP_l=0$. Since this holds for every $k,\,l=0,\,1,\,2,\ldots$, $A=0$. We have shown that the state $\rho_0$ in question is faithful. Particularly, if we fix an orthonormal basis $\{R_n\}_{n=0}^\infty$ of the radial space $L^2(r^2\,dr)$ and numbers $t_l,\,r_n>0$, $l,\,n=0,\,1,\,2,\ldots$, such that $\sum_{l=0}^\infty t_l=1=\sum_{n=0}^\infty r_n$, the state
$$
\rho_0=\sum_{l,n=0}^\infty\sum_{m=-l}^l\frac{t_lr_n}{2l+1}|Y_{l,m}\otimes R_n\>\<Y_{l,m}\otimes R_n|
$$
is faithful and the orbit $\{Z_0(g)\rho_0Z_0(g)^*\,|\,g\in G_0\}$ is homeomorphic with $\R^3$.

The canonical choice for the cyclic and separating vector is now
$$
\Omega=\sum_{l,n=0}^\infty\sum_{m=-l}^l\sqrt{\frac{t_lr_n}{2l+1}}Y_{l,m}\otimes R_n\otimes Y_{l,m}\otimes R_n.
$$
We may define the operation $A\mapsto\ovl A$ with respect to the basis $\{Y_{l,m}\otimes R_n\,|m=0,\ldots,\,l,\ \,l,\,n=0,\,1,\,2,\ldots\}$. Using $\Omega$ in defining the Choi-Jamio\l kowsky isomorphism, one can embark on first characterizing channels covariant with respect to the subgroup $\{(R,0)\,|\,R\in SO(3)\}$ which, according to Theorem \ref{theor:invariant} is equivalent to characterizing the states $S\in\mc S_{\rho_0}(\hil\otimes\mc K)$ which commute with the representation $SO(3)\ni R\mapsto\ovl{Z_0(R,0)}\otimes V(R,0)$. This problem naturally depends on the particular form of $V$. The covariance of the resulting $SO(3)$-covariant channels with respect to $\lambda:\R^3\to\mc U(\hil)$ and $\R^3\ni\vec{a}\mapsto V(\id,\vec{a})\in\mc U(\mc K)$ remains as an additional problem to be solved in order to fully characterize the Euclidean-covariant channels.

\section{Conclusions}

We have established a generalization of the traditional Choi-Jamio\l kowski channel-state dualism which caters for a unital $C^*$-algebra as the Heisenberg input and a $\sigma$-finite injective von Neumann algebra as the Heisenberg output. However, we have concentrated on the case where both input and output are injective von Neumann algebras and the output is, additionally, $\sigma$-finite. We have particularly concentrated on covariant channels and their characterization through their Choi-Jamio\l kowski states. We have seen through examples that the generalized Choi-Jamio\l kwski isomorphism provides effective methods for the study of covariant channels particularly when the symmetry group is compact. In (essentially) the general case, characterizing covariant channels boils down to investigating particular fields of states.

There are plenty of questions regarding covariant channels left for future study: Further research into the case of non-compact symmetry groups remains to be undertaken so that we may successfully tackle the determination of channels covariant under non-compact groups. The first steps taken in Subsection \ref{subsec:euclidean} in characterizing Euclidean-covariant channels seem promising though. In the case of a compact symmetry group $G$, when the input and output algebras are both type-I factors, the problem of characterizing quantum channels $\Phi\in{\bf NCH}(\hil,\mc K)$ such that $\Phi\big(V(g)BV(g)^*\big)=U(g)\Phi(B)U(g)^*$ for all $g\in G$ and $B\in\mc L(\mc K)$, where $U:G\to\mc U(\hil)$ and $V:G\to\mc L(\mc K)$ are some projective unitary representations, is reduced to characterizing states $S\in\mc S(\hil\otimes\mc K)$ with a faithful first margin and $\big(\ovl{U(g)}\otimes V(g)\big)S=S\big(\ovl{U(g)}\otimes V(g)\big)$ for all $g\in G$. Thus the problem of decomposing the tensor product representation $g\mapsto \ovl{U(g)}\otimes V(g)$ into irreducibles (the Clebsch-Gordan problem) becomes crucial. Thus answers to these decomposition problems fully resolve the problem of determining quantum channels covariant with respect to a group that leaves a faithful state invariant.

Many of the problems studied earlier using dilation methods can be examined using the Choi-Jamio\l kowski isomorphism. One such question is {\it incompatibility}: Suppose that $\mc B_1$, $\mc B_2$, and $\mc A$ are unital $C^*$-algebras and $\|\cdot\|_x$ is a cross norm for $\mc B_1\alg\mc B_2$ and denote the $\|\cdot\|_x$-closure of $\mc B_1\alg\mc B_2$ by $\mc B_1\otimes_x\mc B_2$. According to the definitions of \cite{Kuramochi2017}, maps $\Phi_i\in{\bf CH}(\mc A,\mc B_i)$, $i=1,\,2$, are {\it $x$-compatible} if there is a {\it joint channel} $\Psi\in{\bf CH}(\mc A,\mc B_1\otimes_x\mc B_2)$ for $\Phi_1$ and $\Phi_2$, i.e.,\ $\Psi(b_1\otimes 1_{\mc B_2})=\Phi_1(b_1)$ and $\Psi(1_{\mc B_1}\otimes b_2)=\Phi_2(b_2)$ for all $b_i\in\mc B_i$, $i=1,\,2$. Otherwise, the channels $\Phi_1$ and $\Phi_2$ are {\it incompatible}. We may formulate similar compatibility conditions when $\mc A=\mf M$ and $\mc B_i=\mf N_i$, $i=1,\,2$, are von Neumann algebras and we use the von Neumann tensor product $\mf N_1\vN\mf N_2$. Using the Choi-Jamio\l kowski isomorphism, compatibility questions can be studied through reduced states linking compatibility to marginal problems of multipartite states.

\section*{Acknowledgements}

The author would like to thank Dr.\ Juha-Pekka Pellonp\"a\"a, Dr.\ Jukka Kiukas, and Dr.\ Roope Uola for reading earlier versions of the manuscript and for providing their constructive feedback. Also Dr.\ Yui Kuramochi is thanked for pointing out an error in an earlier version of this manuscript. This research has received funding from the National Natural Science
Foundation of China (Grant No. 11875110).


\begin{thebibliography}{99}

\bibitem{AcCe82}
L.\ Accardi and C. Cecchini, ``Conditional Expectations in von Neumann Algebras and a Theorem of Takesaki'', {\it J.\ Funct.\ Anal.}\ {\bf 45}, pp. 245-273 (1982)

\bibitem{Araki74}
H.\ Araki, ``Some properties of modular conjugation operator of von Neumann algebras and a non-commutative Radon-Nikodym theorem with a chain rule'', {\it Pacific J. Math.} {\bf 50}, pp. 309-354 (1974)

\bibitem{kirja}
P.\ Busch, P.\ Lahti, J.-P.\ Pellonp\"a\"a, and K.\ Ylinen,
{\it ``Quantum Measurement''} (Springer-Verlag, 2016)

\bibitem{Dix}
J.\ Dixmier, {\it ``Von Neumann Algebras''} (North-Holland Publishing Company, Amsterdam - New York - Oxford, 1981), English translation of the original {\it ``Les Alg\`ebres d'Op\'erateurs dans l'Espace Hilbertien (Alg\`ebres de Von Neumann), Deuxi\`eme \'Edition''} (Gauthier-Villars 1969)

\bibitem{HaPe2017}
E.\ Haapasalo and J.-P.\ Pellonp\"a\"a, ``Covariant KSGNS construction and quantum instruments'', {\it Rev.\ Math.\ Phys.}\ {\bf 29}, pp. 1-47 (2017)

\bibitem{Holevo2011}
A.\ S.\ Holevo, ``On the Choi-Jamiolkowski Correspondence in Infinite Dimensions'', {\it J.\ Math.\ Phys.}\ {\bf 52}, 042202 (2011)

\bibitem{Jamiolkowski}
A.\ Jamio\l kowski, ``Linear transformations which preserve trace and positive semidefiniteness of operators'' {\it Rep.\ Math.\ Phys.}\ {\bf 3}, pp. 275-278 (1972)

\bibitem{JencovaPetz2006}
A.\ Jencov\'a and D.\ Petz, ``Sufficiency in Quantum Statistical Inference'', {\it Commun.\ Math.\ Phys.}\ {\bf 263}, pp. 259-276 (2006)

\bibitem{KiBuUoPe2017}
J.\ Kiukas, C.\ Budroni, R.\ Uola, and J.-P.\ Pellonp\"a\"a, ``Continuous variable steering and incompatibility via state-channel duality'', {\it Phys.\ Rev.\ A} {\bf 96}, 042331 (2017)

\bibitem{Kuramochi2017}
Y.\ Kuramochi, ``Quantum incompatibility of channels with general outcome operator algebras'', {\it J.\ Math.\ Phys.}\ {\bf 59}, 042203 (2018)

\bibitem{TakesakiI}
M.\ Takesaki, {\it ``Theory of Operator algebras I''}, {\it Encyclopaedia of Mathematical Sciences} {\bf 124} (Springer-Verlag, Berlin Heidelberg, 2002)

\bibitem{TakesakiII}
M.\ Takesaki, {\it ``Theory of Operator algebras II''}, {\it Encyclopaedia of Mathematical Sciences} {\bf 125} (Springer-Verlag, Berlin Heidelberg, 2003)

\bibitem{TakesakiIII}
M.\ Takesaki, {\it ``Theory of Operator algebras III''}, {\it Encyclopaedia of Mathematical Sciences} {\bf 127} (Springer-Verlag, Berlin Heidelberg, 2003)

\bibitem{Varadarajan}
V.\ S.\ Varadarajan, {\it ``Geometry of Quantum Theory - Second Edition''} (Springer-Verlag Berlin 1985)

\end{thebibliography}
\end{document}